\documentclass[sigconf, nonacm, pdfa]{acmart}

\newcommand\vldbdoi{10.14778/3665844.3665854}
\newcommand\vldbpages{2241 - 2254}
\newcommand\vldbvolume{17}
\newcommand\vldbissue{9}
\newcommand\vldbyear{2024}
\newcommand\vldbauthors{\authors}
\newcommand\vldbtitle{\shorttitle} 
\newcommand\vldbavailabilityurl{https://github.com/WeiJiuQi/DET-LSH}
\newcommand\vldbpagestyle{empty} 

\usepackage{caption}
\usepackage{amsmath,amsfonts}
\usepackage{amsthm}
\usepackage{algorithmic}
\usepackage{graphicx}
\usepackage{textcomp}
\usepackage{subfigure}
\usepackage{booktabs}
\usepackage{xcolor}
\usepackage{multirow}
\usepackage{array}
\usepackage{threeparttable}
\usepackage[ruled,vlined]{algorithm2e}
\usepackage{balance}
\usepackage{lipsum}
\usepackage{csquotes}
\usepackage{makecell}

\newtheorem{definition}{Definition}

\newtheorem{lemma}{Lemma}

\newtheorem{theorem}{Theorem}

\newcommand\blfootnote[1]{%
\begingroup
\renewcommand\thefootnote{}\footnote{#1}%
\addtocounter{footnote}{-1}%
\endgroup
}

\begin{document}
\title{DET-LSH: A Locality-Sensitive Hashing Scheme with Dynamic Encoding Tree for Approximate Nearest Neighbor Search}

%%
%% The "author" command and its associated commands are used to define the authors and their affiliations.
\author{Jiuqi Wei}
\affiliation{%
  \institution{Institute of Computing Technology, Chinese Academy of Sciences}
  \institution{University of Chinese Academy of Sciences}
}
\email{weijiuqi19z@ict.ac.cn}

\author{Botao Peng}
\orcid{0000-0002-1825-0097}
\affiliation{%
  \institution{Institute of Computing Technology, Chinese Academy of Sciences}
}
\email{pengbotao@ict.ac.cn}
% \authornote{Corresponding author.}

\author{Xiaodong Lee}
\affiliation{%
  \institution{Institute of Computing Technology, Chinese Academy of Sciences}
}
\email{xl@ict.ac.cn}

\author{Themis Palpanas}
\affiliation{%
  \institution{LIPADE, Universit{\'e} Paris Cit{\'e}}
}
\email{themis@mi.parisdescartes.fr}

% \author{Naghi Mamuli$^*$,\hspace{1em} Arastoo Amel$^{*\mathsection}$,\hspace{1em} Homa Saadat$^*$ }
% \affiliation{%
%  \vspace{-1em}\institution{Aliabad State University$^*$ \hspace{0.3em} North Laboratory$^\mathsection$}
% }
% \affiliation{%
%   \institution{\{naghi,\hspace{0.2em}arastoo,\hspace{0.2em}homa\}@aliabad.ac.ir}
% }

\begin{abstract}
Locality-sensitive hashing (LSH) is a well-known solution for approximate nearest neighbor (ANN) search in high-dimensional spaces due to its robust theoretical guarantee on query accuracy. 
Traditional LSH-based methods mainly focus on improving the efficiency and accuracy of the query phase by designing different query strategies, 
but pay little attention to improving the efficiency of the indexing phase. 
They typically fine-tune existing data-oriented partitioning trees to index data points and support their query strategies.
However, their strategy to directly partition the multi-dimensional space is time-consuming, and performance degrades as the space dimensionality increases.
% However, it is time-consuming for data-oriented partitioning trees to partition directly in a multi-dimensional space, and the performance decreases as the space dimensionality increases.
% However, data-oriented partitioning trees limit the efficiency of methods because partitioning directly in a multi-dimensional space is time-consuming, 
% and the performance of these trees decreases as the space dimensionality increases.
In this paper, we design an encoding-based tree called Dynamic Encoding Tree (DE-Tree) to improve the indexing efficiency and support efficient range queries based on Euclidean distance. 
Based on DE-Tree, we propose a novel LSH scheme called DET-LSH. DET-LSH adopts a novel query strategy, 
which performs range queries in multiple independent index DE-Trees 
to reduce the probability of missing exact NN points, 
thereby improving the query accuracy. 
Our theoretical studies show that DET-LSH enjoys probabilistic guarantees on query accuracy. 
Extensive experiments on real-world datasets demonstrate the superiority of DET-LSH over the state-of-the-art LSH-based methods on both efficiency and accuracy. 
While achieving better query accuracy than competitors, 
DET-LSH achieves up to 6x speedup in indexing time and 2x speedup in query time over the state-of-the-art LSH-based methods. 
\end{abstract}

\maketitle

\blfootnote{$^\ast$ Botao Peng and Xiaodong Lee are the corresponding authors.}

\pagestyle{\vldbpagestyle}
\begingroup\small\noindent\raggedright\textbf{PVLDB Reference Format:}\\
\vldbauthors. \vldbtitle. PVLDB, \vldbvolume(\vldbissue): \vldbpages, \vldbyear.\\
\href{https://doi.org/\vldbdoi}{doi:\vldbdoi}
\endgroup
\begingroup
\renewcommand\thefootnote{}\footnote{\noindent
This work is licensed under the Creative Commons BY-NC-ND 4.0 International License. Visit \url{https://creativecommons.org/licenses/by-nc-nd/4.0/} to view a copy of this license. For any use beyond those covered by this license, obtain permission by emailing \href{mailto:info@vldb.org}{info@vldb.org}. Copyright is held by the owner/author(s). Publication rights licensed to the VLDB Endowment. \\
\raggedright Proceedings of the VLDB Endowment, Vol. \vldbvolume, No. \vldbissue\ %
ISSN 2150-8097. \\
\href{https://doi.org/\vldbdoi}{doi:\vldbdoi} \\
}\addtocounter{footnote}{-1}\endgroup

\ifdefempty{\vldbavailabilityurl}{}{
\vspace{.3cm}
\begingroup\small\noindent\raggedright\textbf{PVLDB Artifact Availability:}\\
The source code, data, and/or other artifacts have been made available at \url{https://github.com/WeiJiuQi/DET-LSH}. 
\endgroup
}
 
\section{Introduction}

\noindent \textbf{Background and Problem.} Nearest neighbor (NN) search in high-dimensional Euclidean spaces is a fundamental problem in various fields, such as database \cite{ferhatosmanoglu2001approximate}, information retrieval \cite{karpukhin2020dense}, data mining \cite{tagami2017annexml}, and machine learning \cite{awale2018polypharmacology}. 
Given a dataset $\mathcal D$ of $n$ data points in $d$-dimensional space $\mathbb{R}^d$ and a query $q$, an NN query returns a point $o^* \in \mathcal D$ which has the minimum Euclidean distance to $q$ among all points in $\mathcal D$. 
However, NN search in high-dimensional datasets is challenging due to the \enquote{curse of dimensionality} phenomenon \cite{hinneburg2000nearest, weber1998quantitative, borodin1999lower}. 
In practice, Approximate Nearest Neighbor (ANN) search is often used as an alternative, 
sacrificing some query accuracy to achieve a huge improvement in efficiency~\cite{lsbforest, fu2016efanna,zeyubulletin-sep23,li2019approximate,annbulletin}. 
Given an approximation ratio $c$ and a query $q \in \mathbb{R}^d$, a $c$-ANN query returns a point $o$ whose distance to $q$ is at most $c$ times the distance between $q$ and its exact NN $o^*$, i.e., $\left\|q,o\right\| \leq c \cdot \left\|q,o^*\right\|$.

\noindent \textbf{Prior Work.} Locality-sensitive hashing (LSH)-based methods are known for their robust theoretical guarantees on the accuracy of query results, 
making them popular in high-dimensional $c$-ANN search \cite{e2lsh, dblsh, c2lsh, qalsh, r2lsh, vhp, lccslsh, srs, pmlsh, lazylsh, eilsh, andoni2015optimal}. 
At the core of LSH-based methods is a family of LSH functions to map points from the original high-dimensional space to low-dimensional projected spaces, and then construct indexes to efficiently support queries, thus reducing the complexity of indexing and querying. 
Thanks to the properties of LSH, points that are close in the original space are more likely to be close in the projected space than those far away~\cite{gionis1999similarity}. 
Therefore, high-quality results can be obtained by only checking the points around the query point in the projected spaces~\cite{datar2004locality}. 
Based on the query strategies, we classify the mainstream LSH-based methods into three categories: 
1) boundary constraint (BC) based methods \cite{e2lsh, lsbforest, sklsh, dblsh}; 
2) collision counting (C2) based methods  \cite{c2lsh, qalsh, r2lsh, vhp, lccslsh}; 
and 3) distance metric (DM) based methods \cite{srs, pmlsh}. 
BC methods map all data points to $L$ independent $K$-dimensional projected spaces, and each projected point is assigned to a hash bucket whose boundary is constrained by a $K$-dimensional hypercube. 
Among $L$ hash tables, two points can be considered \textit{colliding} as long as they are assigned to the same hash bucket at least once. 
Compared with BC methods, which require simultaneous collisions in $K$ dimensions, C2 methods relax the collision condition. 
C2 methods select candidate points whose number of collisions with the query point is greater than a predefined threshold. 
In DM methods, the distance between two points in the projected space can be used to estimate their distance in the original space with theoretical guarantees. 
Therefore, DM methods select candidate points by conducting range queries based on the Euclidean distance metric in the projected space.

\noindent \textbf{Limitations and Motivation.} Efficiency and accuracy are key metrics to evaluate the performance of LSH-based methods in $c$-ANN search. 
% Efficiency is reflected in both the indexing and query phases of a method, i.e., the indexing time and the query time. 
% Accuracy is typically measured by evaluating the difference between the returned ANN points and the exact NN points. 
Nowadays, new data is produced at an ever-increasing rate, and the size of datasets is continuously growing \cite{DBLP:journals/sigmod/Palpanas15,Palpanas2019,fernandez2020data,wei2023data}. We need to manage large-scale data more efficiently to support further data analysis \cite{wellenzohn2023robust,peng2023efficient,peng2022lan,hydra2}.
However, existing LSH-based methods mainly focus on reducing query time and improving query accuracy by designing different query strategies, but pay little attention to reducing indexing time \cite{dblsh, qalsh, r2lsh, vhp, srs, pmlsh}. 
They typically fine-tune existing data-oriented partitioning trees to index data points and support their query strategies, such as R*-Tree \cite{rstartree} for DB-LSH \cite{dblsh}, PM-Tree \cite{pmtree} for PM-LSH \cite{pmlsh}, and R-Tree \cite{rtree} for SRS \cite{srs}. 
Data-oriented partitioning trees \cite{rtree, mtree, rstartree, pmtree} group nearby data points and partition them into their minimum bounding objects (e.g., hyperrectangle, hypersphere) hierarchically. 
% Figure \ref{Partition-based Trees} illustrates the partitioning process and index structure of an R-Tree. 
However, partitioning directly in a multi-dimensional space is time-consuming, which limits the efficiency of these methods in the indexing phase. 
In addition, the performance of data-oriented partitioning trees decreases as the space dimensionality increases \cite{bohm2000cost, weber1998quantitative}, which limits the dimensionality of the projected space. 
Therefore, it is necessary to design a more efficient tree structure to address the limitations. 
From another perspective, a more efficient tree structure can also help improve query accuracy, since more trees can be constructed in the same indexing time, and query answering based on more trees can be more accurate. 
For example, the state-of-the-art method among BC methods, DB-LSH~\cite{dblsh}, constructs five R*-Trees to reduce the probability of missing exact NN points in the query phase. 
% thereby improving the accuracy of the results.
%\textcolor{red}{The graph-based methods~\cite{lshapg,malkov2014approximate,malkov2018efficient,munoz2019hierarchical} improved response speed. However, the slow graph construction cannot meet data growth. Therefore, we focus on tree-based ANN implementation. }
% We focus on tree-based LSH methods, since they exhibit fast index construction times~\cite{hydra2}. 

\noindent \textbf{Our Method.} In this paper, we propose a novel tree structure called Dynamic Encoding Tree (DE-Tree) and a novel LSH scheme called DET-LSH to solve the high-dimensional $c$-ANN search problem more efficiently and accurately. 
First, we present an encoding-based tree called DE-Tree, which divides and encodes each dimension of the projected space independently (as shown in Figure~\ref{Encoding-based Trees}), avoiding to directly partition the multi-dimensional projected space like data-oriented partitioning trees do. 
This idea leads to improved indexing efficiency. 
DE-Tree dynamically encodes projected points based on the dataset's distribution, so that nearby points have more similar encoding representations than distant ones, thereby improving query accuracy.
DE-Tree supports efficient range queries because the upper and lower bound distances between a query point and any DE-Tree node can be easily calculated.
Second, we propose a novel LSH scheme called DET-LSH. 
DET-LSH dynamically encodes $K$-dimensional projected points and then constructs $L$ DE-Trees based on the encoded points. 
We design a two-step query strategy for DET-LSH, which combines the ideas of BC and DM methods.
The first step is to perform range queries in DE-Trees and identify in a coarse-grained way a certain proportion of candidate points that are close to the query point. 
The second step is to calculate the actual distance of each candidate point from the query point in a fine-grained way, then sort the distances and return the final result. 
Intuitively, the coarse-grained filtering improves the query efficiency, and the fine-grained calculation improves query accuracy.
Third, we conduct a rigorous theoretical analysis showing that DET-LSH can correctly answer a $c^2$-$k$-ANN query with a constant probability. 
Furthermore, extensive experiments demonstrate that DET-LSH outperforms existing LSH-based methods in both efficiency and accuracy.

Our main contributions are summarized as follows.

\begin{itemize}
	\item We present a novel encoding-based tree structure called DE-Tree. Compared with data-oriented partitioning trees used in existing LSH-based methods, DE-Tree has better indexing efficiency and can support more efficient range queries based on the Euclidean distance metric.
	\item We propose DET-LSH, a novel LSH scheme based on DE-Tree. We design a novel query strategy for DET-LSH, taking into account both efficiency and accuracy. We provide a theoretical analysis showing that DET-LSH answers a $c^2$-$k$-ANN query with a constant success probability.
	\item We conduct extensive experiments, demonstrating that DET-LSH can achieve better efficiency and accuracy than existing LSH-based methods. While achieving better query accuracy than competitors, DET-LSH achieves up to 6x speedup in indexing time and 2x speedup in query time over the state-of-the-art LSH-based methods.
\end{itemize}

\section{Related Work} \label{chapter2}

\subsection{Mainstream LSH-based Methods}  \label{chapter2.1}

\noindent\textbf{Boundary Constraint based methods (BC).} BC requires $K \cdot L$ hash functions to map all data points to $L$ independent $K$-dimensional projected spaces. 
Each projected point is assigned to a hash bucket whose boundary is constrained by a $K$-dimensional hypercube. 
Among $L$ hash tables, two points can be considered colliding as long as they are assigned to the same hash bucket at least once. 
E2LSH \cite{e2lsh} is the original BC method that adopts LSH functions following the $p$-stable distribution \cite{datar2004locality}. 
E2LSH needs to continuously generate new hash tables when the search radius $r$ gradually increases, which leads to prohibitively large space consumption in indexing. 
To alleviate this issue, LSB-Forest \cite{lsbforest} adopts B-Tree \cite{btree} to index projected points, avoiding building hash tables at different radii. 
SK-LSH \cite{sklsh} proposes a novel index structure based on B$^+$-Tree \cite{btree}, and the search strategy supports it finding better candidates with lower I/O cost. 
However, neither LSB-Forest nor SK-LSH ensures any LSH-like theoretical guarantees since they are based on heuristics. 
DB-LSH \cite{dblsh} is the state-of-the-art BC method with strict theoretical guarantees, which presents a dynamic search framework based on R$^*$-Tree \cite{rstartree}.
% DB-LSH \cite{dblsh} is the state-of-the-art BC method with strict theoretical guarantees. 
% DB-LSH designs a dynamic search framework based on R$^*$-Tree \cite{rstartree}, which associates the search radius $r$ with the boundary width $w$ of hash buckets.

% \begin{figure} [t!]
% 	\centering 
% 	\subfigcapskip=5pt
% 	\subfigure[Partition points into minimum bounding hyperrectangles.]{
% 		\includegraphics[width=0.73\linewidth]{./figures/rtreepartitioning.pdf}
% 		\label{space partitioning}}
% 	\subfigure[An R-Tree index based on the partitions.]{
% 		\includegraphics[width=0.55\linewidth]{./figures/rtreeindex.pdf}
% 		\label{R-Tree index}}
% 	\caption{\textcolor{red}{An illustration of data-oriented partitioning tree.}}
% 	\label{Partition-based Trees}
% \end{figure}

\noindent\textbf{Collision Counting based methods (C2).} C2 requires $K^{'} \cdot L^{'}$ hash functions to construct $L^{'}$ independent $K^{'}$-dimensional hash  tables, where $K^{'}<K$ and $L^{'} > L$. 
C2 selects candidate points whose number of collisions is greater than a threshold $t$, where $t < L^{'}$. C2LSH \cite{c2lsh} proposes the C2 scheme and only maintain $K^{'}$ one-dimensional hash tables ($K^{'}=1$). 
C2LSH adopts the \textit{virtual rehashing} technique to count collisions dynamically, reducing index space consumption. 
QALSH \cite{qalsh} improves C2LSH by using B$^+$-Trees to locate points projected into the same bucket, avoiding counting the collision numbers among a large number of points dimension by dimension. 
% QALSH also proposes the query-aware dynamic bucketing technique, which alleviates the hash boundary issue and supports virtual rehashing. 
To further reduce the space consumption of QALSH, R2LSH \cite{r2lsh} and VHP \cite{vhp} are proposed. 
R2LSH maps data points into multiple two-dimensional projected spaces ($K^{'}=2$) and VHP considers hash buckets as virtual hypersphere ($K^{'}>2$). 
% Both of them construct B$^+$-Trees to index and query projected points.
% LCCS-LSH \cite{lccslsh} proposes a novel search framework based on the longest circular co-substring, 
LCCS-LSH \cite{lccslsh} proposes a novel search framework, 
which extends C2's method of counting collisions from the number of discrete points to the length of continuous co-substrings.

\noindent\textbf{Distance Metric based methods (DM).} The intuition of DM is that the points close to query $q$ in the original space are also close to query $q$ in the projected space. 
DM requires $K$ hash functions to map data points into a $K$-dimensional projected space. 
SRS~\cite{srs} utilizes R-tree to index projected points and performs exact NN search in the $K$-dimensional projected space. PM-LSH~\cite{pmlsh} designs a range query mechanism based on PM-Tree~\cite{pmtree} to improve query efficiency. 
According to Euclidean distances between queries and points in the projected space, $\beta n+k$ candidates will be selected in PM-LSH, where $\beta$ is an estimated ratio to guarantee the ANN search performance and $n$ is the dataset cardinality.

\subsection{Tree Structure}  \label{chapter2.3}

\noindent\textbf{Data-oriented Partitioning Tree.} As mentioned above, mainstream LSH-based methods adopt data-oriented partitioning trees, such as B-Tree~\cite{btree}, R-Tree~\cite{rtree}, M-Tree~\cite{mtree}, and their variants~\cite{rstartree, pmtree}, to construct indexes and support queries. 
In these methods, data-oriented partitioning trees group nearby data points and hierarchically partition them into their minimum bounding graphics (e.g., hyperrectangle, hypersphere). 
For example, R-Tree and M-Tree partition data points into hyperrectangular and hyperspherical partitions, respectively. 
% Figure \ref{Partition-based Trees} illustrates the partitioning process and index structure of an R-Tree. 
However, for LSH-based methods, partitioning multi-dimensional projected spaces consumes much time. In addition, with the increase of space dimensionality, the effectiveness of data-oriented partitioning trees decreases~\cite{bohm2000cost, weber1998quantitative}, 
which is also the reason why tree-based methods~\cite{vptree, balltree, rptree, optimisedkdtree} cannot efficiently support ANN search in high-dimensional spaces.

\noindent\textbf{Encoding-based Tree.} Encoding-based trees play an important role in data series similarity search \cite{iSAX2,wang2013data, ads, messi, coconut, peng2018paris,paris+,DBLP:journals/vldb/PengFP21, ulisse, sing, chatzakis2023odyssey, fatourou2023fresh, seanet, echihabi2022hercules, dumpy}. 
Unlike data-oriented partitioning trees, which index a data point directly based on their multi-dimensional coordinates, encoding-based trees independently encode the coordinates of each dimension of the data point into symbolic representations. 
The indexable Symbolic Aggregate approXimation (iSAX)~\cite{isax} is a widely used symbolic representation. 
iSAX divides each dimension into non-uniformly distributed regions and assigns a bit-wise symbol to each region. 
Figure~\ref{isaxencoding} illustrates the encoding process for iSAX representations, 
and Figure~\ref{isaxindex} illustrates an iSAX index based on the representations. 
In practice, iSAX requires only 256 symbols for a very good approximation, so the maximum alphabet cardinality can be represented by 8 bits~\cite{iSAX2}. 
Based on the iSAX representation, several encoding-based trees with different indexing and query strategies are proposed to support data series similarity search \cite{iSAX2, ads, messi, paris+, sing, chatzakis2023odyssey, fatourou2023fresh, seanet, dumpy}. 
The advantages of encoding-based trees can be transferred to LSH-based methods for ANN search. 
Specifically, encoding-based trees divide and encode each dimension of the space independently, avoiding partitioning multi-dimensional projected spaces, improving indexing efficiency. 
In addition, the upper and lower bound distances between two points can be calculated easily using their region boundaries, 
which is suitable for range queries in LSH-based methods, improving query efficiency.

\begin{figure} [t!]
	\centering
	\subfigcapskip=5pt
	\subfigure[Encode data points into iSAX representations.]{
		\includegraphics[width=0.8\linewidth]{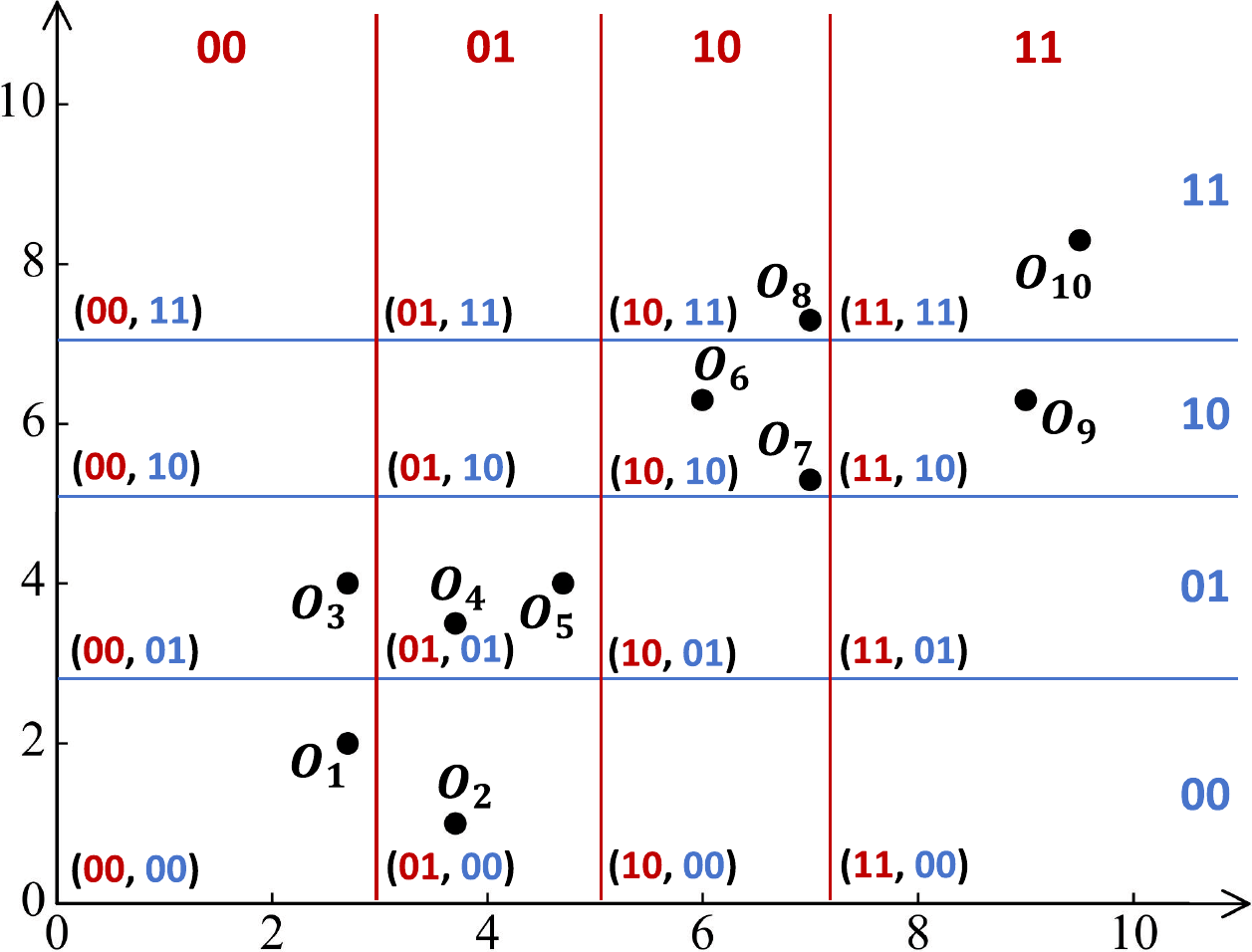}
		\label{isaxencoding}}
	\subfigure[An index based on the iSAX representations.]{
		\includegraphics[width=0.8\linewidth]{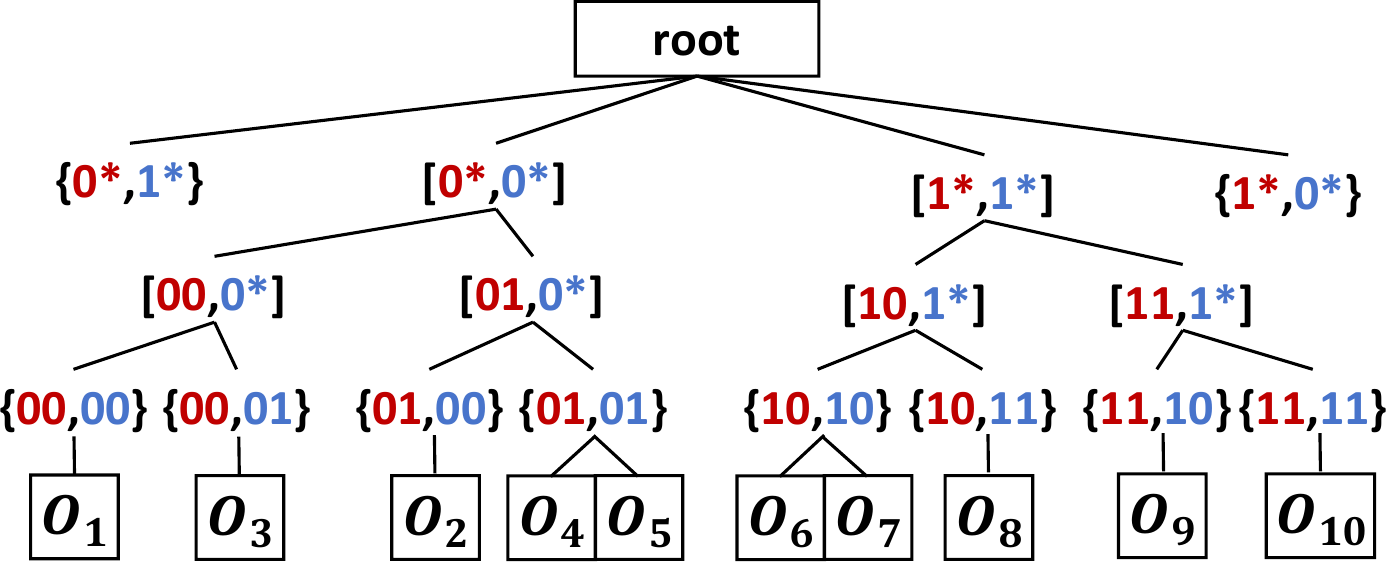}
		\label{isaxindex}}
	\caption{Illustration of an encoding-based tree.}
	\label{Encoding-based Trees}
\end{figure}

\section{Preliminaries} \label{chapter3}

\subsection{Problem Definition}  \label{chapter3.1}

Let $\mathcal D$ be a dataset of points in $d$-dimensional space $\mathbb{R}^d$. 
The dataset cardinality is denoted as $\lvert \mathcal D \rvert=n$, 
and let $\left\|o_1,o_2\right\|$ denote the distance between points $o_1,o_2\in \mathcal D$. 
The query point $q \in \mathbb{R}^d$.

\begin{definition}[$c$-ANN]\label{def1}
	Given a query point $q$ and an approximation ratio $c > 1$, let $o^*$ be the exact nearest neighbor of $q$ in $\mathcal D$. A $c$-ANN query returns a point $o \in \mathcal D$ satisfying $\left\|q,o\right\| \leq c \cdot \left\|q,o^*\right\|$.
\end{definition} 

The $c$-ANN query can be generalized to $c$-$k$-ANN query that returns $k$ approximate nearest points, where $k$ is a positive integer.

\begin{definition}[$c$-$k$-ANN]\label{def2}
	Given a query point $q$, an approximation ratio $c > 1$, and an integer $k$. Let $o^*_i$ be the $i$-th exact nearest neighbor of $q$ in $\mathcal D$. A $c$-$k$-ANN query returns $k$ points $o_1,o_2,...,o_k$. For each $o_i \in D$ satisfying $\left\|q,o_i\right\| \leq c \cdot \left\|q,o^*_i\right\|$, where $i \in [1,k]$.
\end{definition}  

In fact, LSH-based methods do not solve $c$-ANN queries directly because $o^*$ and $\left\|q,o^*\right\|$ is not known in advance \cite{lccslsh, pmlsh, dblsh}. 
Instead, they solve the problem of ($r$,$c$)-ANN proposed in \cite{indyk1998approximate}.
% Instead, they solve the problem of ($r$,$c$)-ANN proposed in \cite{indyk1998approximate}, which is a decision version of $c$-ANN.

\begin{definition}[($r$,$c$)-ANN]\label{def3}
	Given a query point $q$, an approximation ratio $c > 1$, and a search radius $r$. An ($r$,$c$)-ANN query returns the following result:
	\begin{enumerate}
		\item If there exists a point $o \in \mathcal D$ such that $\left\|q,o\right\| \leq r$, then return a point $o^{\prime} \in \mathcal D$ such that $\left\|q,o^{\prime}\right\| \leq c \cdot r$;
		\item If for all $o \in \mathcal D$ we have $\left\|q,o\right\| > c \cdot r$, then return nothing;
		\item If for the point $o$ closest to $q$ we have $r < \left\|q,o\right\| \leq c \cdot r$, then return $o$ or nothing. 
	\end{enumerate}
\end{definition}

\begin{table}
	\centering
	\caption{Notations}
	\label{table1}
 {\small
	\begin{tabular}{cc}
		\toprule
		\textbf{Notation} & \textbf{Description} \\
		\midrule
		$\mathbb{R}^d$ & $d$-dimensional Euclidean space \\
		$\mathcal D, n$ & Dataset of points in $\mathbb{R}^d$ and its cardinality  $\lvert \mathcal D \rvert$ \\
		$o, q$ & A data point in $\mathcal D$ and a query point in $\mathbb{R}^d$ \\
		$o^{\prime}, q^{\prime}$ & $o$ and $q$ in the projected space \\
		$o^*, o^*_i$ & The first and $i$-th nearest point in $\mathcal D$ to $q$ \\
		$\left\|o_1,o_2\right\|$ & The Euclidean distance between $o_1$ and $o_2$ \\
		$s, s^\prime$ & Abbreviation for $\left\|o_1,o_2\right\|$ and  $\left\|o_1^\prime,o_2^\prime \right\|$ \\
		$h(o)$ & Hash function \\
		$\mathcal H (o)$ & $[h_1(o),...,h_K(o)]$, the coordinates of $o^{\prime}$ \\
		$\mathcal H_i(o)$ & Coordinates of $o^{\prime}$ in the $i$-th project space\\
		$c$ & Approximation ratio \\
		$r, r_{min}$ & Search radius and the initial search radius \\
		$d, K$ & Dimension of the original and the projected space \\
		$L$ & Number of independent projected spaces \\
		\bottomrule
	\end{tabular}
 } % font size
\end{table}

The $c$-ANN query can be transformed into a series of ($r$,$c$)-ANN queries with increasing radii until a point is returned. 
The search radius $r$ is continuously enlarged by multiplying $c$, i.e., $r=r_{min}, r_{min} \cdot c, r_{min} \cdot c^2,...$, where $r_{min}$ is the initial search radius. In this way, as proven by \cite{indyk1998approximate}, the ANN query can be answered with an approximation ratio $c^2$, i.e., $c^2$-ANN.

\begin{figure*} [tb]
	\flushleft 
	\subfigcapskip=5pt
	\subfigure[Encoding phase and indexing phase.]{
		\includegraphics[width=0.48\linewidth]{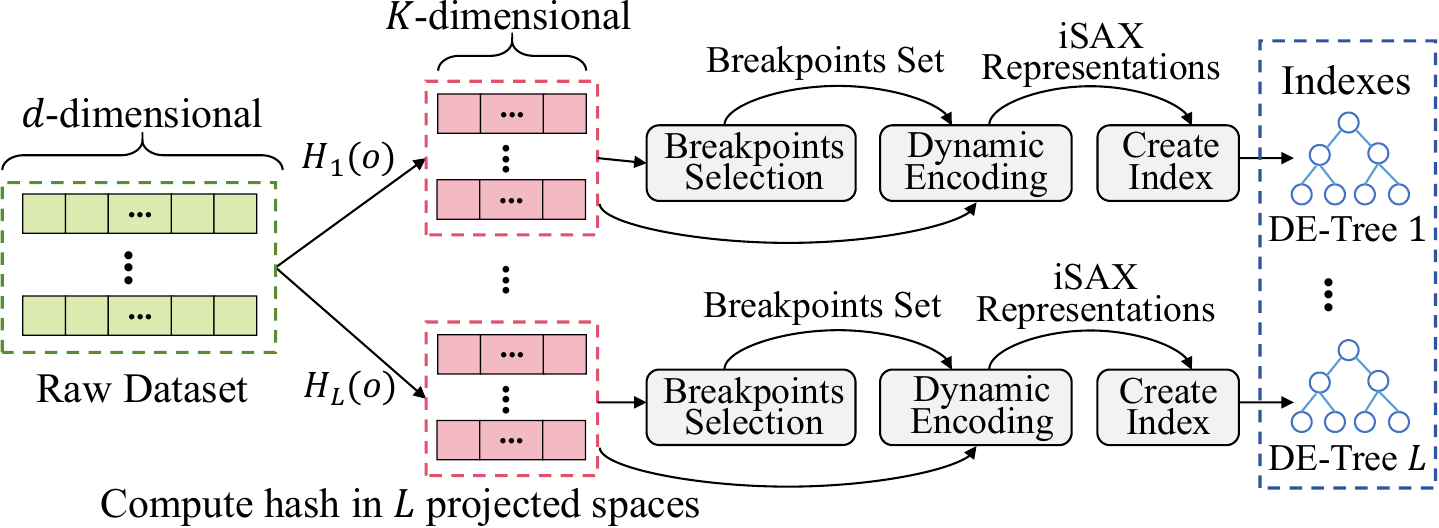}
		\label{overview1}}\hspace{3mm}
	\subfigure[Query phase.]{
		\includegraphics[width=0.48\linewidth]{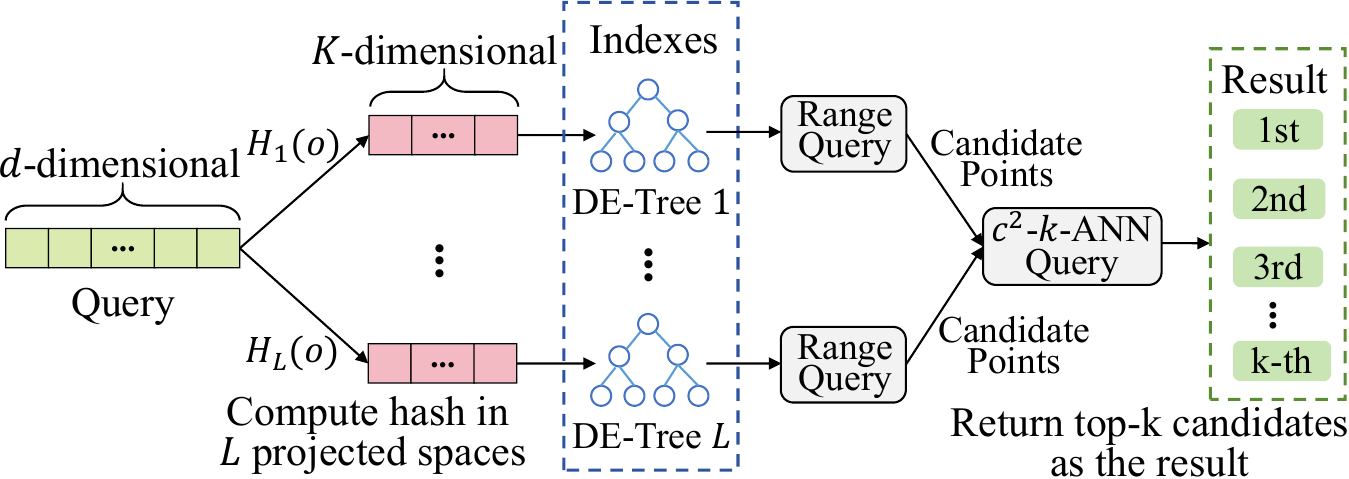}
		\label{overview2}}
	\caption{Overview of the DET-LSH workflow.}
	\label{overview}
\end{figure*}

\subsection{Locality-Sensitive Hashing}  \label{chapter3.2}

The capability of an LSH function $h$ is to project closer data points into the same hash bucket with a higher probability. Formally, the definition of LSH used in Euclidean space is given below \cite{pmlsh, dblsh}: 

\begin{definition}[LSH]\label{def4}
	Given a distance $r$, an approximation ratio $c>1$, a family of hash functions $\mathcal H = \{h:\mathbb{R}^d \rightarrow \mathbb{R}\}$ is called ($r$,$cr$,$p_1$,$p_2$)-locality-sensitive, if for $\forall o_1,o_2 \in \mathbb{R}^d$, it satisfies both of the following conditions: 
	\begin{enumerate}
	   	\item If $\left\|o_1,o_2\right\| \leq r$, $\Pr{[h(o_1)=h(o_2)]} \geq p_1$;
	   	\item If $\left\|o_1,o_2\right\| > cr$, $\Pr{[h(o_1)=h(o_2)]} \leq p_2$,
	\end{enumerate}
	where $h \in \mathcal H$ is randomly chosen, and the probability values $p_1$ and $p_2$ satisfy $p_1>p_2$.
\end{definition} 

A widely adopted LSH family for the Euclidean space is defined as follows \cite{qalsh}:

\begin{equation}  \label{eq1}
	h(o)=\vec{a} \cdot \vec{o},
\end{equation}

\noindent where $\vec{o}$ is the vector representation of a point $o \in \mathbb{R}^d$ and $\vec{a}$ is a $d$-dimensional vector where each entry is independently chosen from the standard normal distribution $\mathcal N(0,1)$.

\subsection{$p$-Stable Distribution and $\chi^2$ Distribution}  \label{chapter3.3}

A distribution $\mathcal T$ is called $p$-stable, if for any $u$ real numbers $v_1,...,v_u$ and identically distributed (i.i.d.) variables $X_1,...,X_u$ following $\mathcal T$ distribution, $\sum_{i=1}^{u}v_iX_i$ has the same distribution as $(\sum_{i=1}^{u}\lvert v_i \rvert ^p)^{1/p} \cdot X$, where $X$ is a random variable with distribution $\mathcal T$ \cite{datar2004locality}. $p$-stable distribution exists for any $p \in (0,2]$ \cite{zolotarev1986one}, and $\mathcal T$ is the normal distribution when $p=2$. 

Let $o^{\prime}= \mathcal H (o)=  [h_1(o),...,h_K(o)]$ denote the point $o$ in the $K$-dimensional projected space. For any two points $o_1,o_2 \in \mathcal D$, let $s=\left\|o_1,o_2\right\|$ and $s^\prime=\left\|o_1^\prime,o_2^\prime\right\|$ denote the Euclidean distances between $o_1$ and $o_2$ in the original space and in the projected space. 

\begin{lemma}\label{lemma1}
	$\frac{s^{\prime2}}{s^2}$ follows the $\chi^2(K)$ distribution.
\end{lemma} 

\begin{proof}
	Let $h^\prime=h(o_1)-h(o_2)=\vec{a} \cdot (\vec{o_1}-\vec{o_2})=\sum_{i=1}^{d}(o_1[i]-o_2[i]) \cdot a[i]$, where $a[i]$ follows the $\mathcal N(0,1)$ distribution. Since 2-stable distribution is the normal distribution, $h^\prime$ has the same distribution as $(\sum_{i=1}^{d} (o_1[i]-o_2[i])^2)^{1/2} \cdot X=s \cdot X$, where $X$ is a random variable with distribution $\mathcal N(0,1)$. Therefore $\frac{h^\prime}{s}$ follows the $\mathcal N(0,1)$ distribution. Given $K$ hash functions $h_1(\cdot),...,h_K(\cdot)$, we have $\frac{h_1^{\prime2}+...+h_K^{\prime2}}{s^2}=\frac{s^{\prime2}}{s^2}$, which has the same distribution as $\sum_{i=1}^{K}X_i^2$. Thus, $\frac{s^{\prime2}}{s^2}$ follows the $\chi^2(K)$ distribution.
\end{proof}

\begin{lemma}\label{lemma2}
	Given $s$ and $s^\prime$ we have:
%	\begin{enumerate}
%		\item $\Pr{[s^\prime>s\sqrt{\chi^2_{\alpha_1}(K)}]} = \alpha_1$;
%		\item $\Pr{[s^\prime<s\sqrt{\chi^2_{1-\alpha_2}(K)}]} = \alpha_2$,
%	\end{enumerate}
	\begin{equation}  \label{eqkafang}
		\Pr{[s^\prime>s\sqrt{\chi^2_{\alpha}(K)}]} = \alpha,
	\end{equation}
	\noindent where $\chi^2_\alpha(K)$ is the upper quantile of a distribution $Y \sim \chi^2(K)$, i.e., $\Pr{[Y > \chi^2_\alpha(K)]} = \alpha$.
\end{lemma} 

\begin{proof}
%	From Lemma \ref{lemma1}, we have 	$\frac{s^{\prime2}}{s^2} \sim \chi^2(K)$. Since $\chi^2_{\alpha_1}(K)$ and $\chi^2_{1-\alpha_2}(K)$ are the $\alpha_1$ and $1-\alpha_2$ upper quantiles of the $\chi^2(K)$ distribution, respectively, we have $\Pr{[ \frac{s^{\prime2}}{s^2} > \chi^2_{\alpha_1}(K)]} = \alpha_1$ and $\Pr{[ \frac{s^{\prime2}}{s^2} > \chi^2_{1-\alpha_2}(K)]} = 1-\alpha_2$. Transform the formulas, we have $\Pr{[s^\prime>s\sqrt{\chi^2_{\alpha_1}(K)}]} = \alpha_1$ and $\Pr{[s^\prime<s\sqrt{\chi^2_{1-\alpha_2}(K)}]} = \alpha_2$.
	From Lemma \ref{lemma1}, we have 	$\frac{s^{\prime2}}{s^2} \sim \chi^2(K)$. Since $\chi^2_{\alpha}(K)$ is the $\alpha$ upper quantiles of $\chi^2(K)$ distribution, we have $\Pr{[ \frac{s^{\prime2}}{s^2} > \chi^2_{\alpha}(K)]} = \alpha$. Transform the formulas, we have $\Pr{[s^\prime>s\sqrt{\chi^2_{\alpha}(K)}]} = \alpha$.
\end{proof}

\section{The DET-LSH Method} \label{chapter4}

In this section, we present the details of DET-LSH and the design of Dynamic Encoding Tree (DE-Tree). 
DET-LSH consists of three phases: an encoding phase to encode the LSH-based projected points into iSAX representations; 
an indexing phase to construct DE-Trees based on the iSAX representations; 
a query phase to perform range queries in DE-Trees for ANN search.
Figure \ref{overview} provides a high-level overview of the workflow for DET-LSH. 
% In the encoding phase, (1) first, $K \cdot L$ hash functions are used to map the raw data to $L$ independent $K$-dimensional projected spaces. 
% (2) Then, dynamically select breakpoints for each dimension based on the dataset’s distribution, 
% thereby dividing each projected space into different regions and assigning a bit-wise symbol to each region. 
% (3) Finally, determine which region the projected points fall and encode them into iSAX representations. 
% In the indexing phase, based on the obtained iSAX representation of all data points, 
% a DE-Tree is constructed in each projected space by splitting the tree nodes in a top-down manner. 
% In the query phase, (1) first, map the query point to $L$ projected spaces. 
% (2) Then, perform range queries on multiple DE-Trees based on the upper and lower bound distance calculation characteristics of the DE-Tree. 
% A candidate set of data points is obtained. 
% (3) Finally, the $c^2$-$k$-ANN query is supported by our designed query strategy, and $top$-$k$ points are returned.

% In this section, we present the details of DET-LSH, and the design of Dynamic Encoding Tree (DE-Tree) is also introduced in this process. DET-LSH consists of three phases: an encoding phase to encode the LSH-based projected points into iSAX representations; an indexing phase to construct DE-Trees based on the iSAX representations; a query phase to perform range queries in DE-Trees for ANN search.

\begin{algorithm}[tb]
% \footnotesize
	\caption{Breakpoints Selection}%算法名字                                                                           
	\label{getbreakpoints}
	\LinesNumbered %要求显示行号
	\KwIn{Parameters $K$, $L$, $n$, all points in projected spaces $P$, sample size $n_s$, number of regions in each projected space $N_r$}%输入参数
	\KwOut{A set of breakpoints $B$}%输出
	Initialize $B$ with size $L \cdot K \cdot (N_r+1)$; \\
	\For{$i=1$ to $L$}{
		\For{$j=1$ to $K$}{
			Sample $C_{ij}=[h_{ij}(o_1),...,h_{ij}(o_{n_s})]$ from $P$; \\
			$round \leftarrow \log_2 N_r$; \\
			\For{$z=1$ to $round$}{
				Use \textit{QuickSelect} algorithm and \textit{divide-and-conquer} strategy to find $2^{z-1}$ breakpoints in round $z$; \\
                Store the found breakpoints in $B_{ij}$; \\
			}
			$final\_region\_size \leftarrow \lfloor \frac{n_s}{2^{round}} \rfloor$; \\
			$B_{ij}(1) \leftarrow$ the minimum element from $C_{ij}(1)$ to $C_{ij}(final\_region\_size)$; \\
			$B_{ij}(N_r+1) \leftarrow$ the maximum element from $C_{ij}(n_s-final\_region\_size)$ to $C_{ij}(n_s)$; \\
		}
	}
	\Return $B$; \\
\end{algorithm}

\subsection{Encoding Phase} \label{Encoding Phase}

DET-LSH first encodes projected points into iSAX representations. iSAX uses \emph{breakpoints} to divide each dimension into non-uniform regions, and assigns a bit-wise symbol to each region. 
For example, Figure \ref{isaxencoding} illustrates an iSAX-based encoding process under a two-dimensional space. 
In Figure \ref{isaxencoding}, we use three breakpoints in each dimension to divide it into four regions, each of which can be represented by a 2-bit symbol: 00/01/10/11. 
Therefore, the space is divided into 16 regions, and the points in the same region have the same iSAX representations. 
% For example, both $O_4$ and $O_5$ are encoded as (01,01). 
Figure \ref{isaxindex} shows an index based on the iSAX representations. 
% where \enquote{$\ast$} indicates that the space can be further divided with more bits. 
In practice, iSAX only requires 256 symbols in each dimension to get a very good approximation \cite{iSAX2}, which means each dimension can be encoded with an 8-bit alphabet.

\noindent \textbf{Static encoding scheme.} In data series similarity search, traditional iSAX-based methods adopt the static encoding scheme \cite{iSAX2, ads, messi, paris+, sing, chatzakis2023odyssey, fatourou2023fresh}. 
Since normalized data series have highly Gaussian distribution \cite{isax}, they simply determine the breakpoints $b_1,...,b_{a-1}$ such that the area under a $\mathcal N(0,1)$ Gaussian curve from $b_i$ to $b_{i+1}$ is $\frac{1}{a}$, where $b_0$ and $b_a$ are defined as $-\infty$ and $+\infty$. 
Therefore, these breakpoints are static and independent of datasets. 
Existing methods encode a data series by checking which two breakpoints each of its coordinates falls between in a common statistical table. 
However, the datasets for ANN search have arbitrary distributions, so the static encoding scheme is no longer suitable.

\noindent \textbf{Dynamic encoding scheme.} In DET-LSH, we design a dynamic encoding scheme to dynamically select breakpoints based on the distribution of the dataset, aiming to divide data points into different regions as evenly as possible, i.e., each region contains the same number of points. 
Specifically, assuming we have a dataset with cardinality $n$, we first use $K \cdot L$ hash functions to calculate the $K$-dimensional points in $L$ projected spaces, where $\mathcal H_i(o)=[h_{i1}(o),...,h_{iK}(o)]$ denote a point $o$ in the $i$-th projected space. 
% This step has $\mathcal{O}(L \cdot K \cdot n \cdot d)$ time cost.
Let $ C_{ij}=[h_{ij}(o_1),...,h_{ij}(o_n)]$ denote the set of coordinates of all $n$ points in their $i$-th projected space and $j$-th dimension, where $i=1,...,L$ and $j=1,...,K$. 
We denote $C^{\uparrow}_{ij}$ as a new set in which the elements of $C_{ij}$ are sorted in ascending order, 
and use $C^{\uparrow}_{ij}(t)$ to represent the $t$-th element in $C^{\uparrow}_{ij}$. 
Intuitively, to make points evenly divided into $N_r=256$ regions in each dimension, we can select ordered breakpoints $B_{ij}$ from $C^{\uparrow}_{ij}$, where $B_{ij}(z)=C^{\uparrow}_{ij}(\lfloor \frac{n}{N_r} \rfloor \cdot (z-1))$ and $z=2,...,N_r$. We set $B_{ij}(1)=C^{\uparrow}_{ij}(1)$ and $B_{ij}(N_r+1)=C^{\uparrow}_{ij}(n)$. 
In practice, we dynamically select $N_r+1$ breakpoints for each dimension. 
Then, for any point $o$, each dimension $h_{ij}(o)$ can be independently encoded based on the selected breakpoints $B_{ij}$.

\begin{algorithm}[tb]
% \small

	\caption{Dynamic Encoding}%算法名字                                                                         
	\label{dynamic_encoding}
	\LinesNumbered %要求显示行号
	\KwIn{Parameters $K$, $L$, $n$, all points in projected spaces $P$, sample size $n_s$, number of regions in each projected space $N_r$}%输入参数
	\KwOut{A set of encoded points $EP$}%输出
	Initialize $EP$ with size $n \cdot L \cdot K$; \\
	$B \leftarrow$ \textbf{call} BreakpointsSelection($K$,$L$,$n$,$P$,$n_s$,$N_r$); \\
	\For{$i=1$ to $L$}{
		\For{$j=1$ to $K$}{
			\For{$z=1$ to $n$}{
                Obtain $o_z$ from $P$; \\
				Use \textit{BinarySearch} to find integer $b \in [1,N_r]$ such that $B_{ij}(b) \leq h_{ij}(o_z) \leq B_{ij}(b+1)$; \\
				$EP_{ij}(o_z) \leftarrow b$-th symbol in the 8-bit alphabet; \\			
			}
		}
	}
	\Return $EP$; \\
\end{algorithm}

In terms of algorithm design, the intuitive idea is to completely sort $C_{ij}$ to get the exact $C^{\uparrow}_{ij}$ and then select breakpoints from it. 
However, we only need $N_r+1$ discrete elements in $C^{\uparrow}_{ij}$, the complete sorting of $C_{ij}$ is wasteful. 
Therefore, combining the \textit{QuickSelect} algorithm with the \textit{divide-and-conquer} strategy, we design a dynamic encoding scheme based on the unordered $C_{ij}$. 
Algorithm \ref{getbreakpoints} introduces how we dynamically select breakpoints, which is the first step of the encoding scheme. 
To improve efficiency, we randomly sample $n_s$ points from the dataset and select breakpoints based on these sampled points. 
In practice, we set $n_s=0.1n$. 
For each $C_{ij}$, Algorithm \ref{getbreakpoints} obtains breakpoints by running multiple rounds of the \textit{QuickSelect} algorithm combined with the \textit{divide-and-conquer} strategy (lines 6-8). 
% To facilitate understanding, we illustrate this process in Figure \ref{breakpoints}. 
For unordered $C_{ij}$, \textit{QuickSelect}($start$, $q$, $end$) can find the $q$-th smallest element between $C_{ij}(start)$ and $C_{ij}(end)$, and move it to the position of $C_{ij}(start+q)$. 
Then, $C_{ij}(start+q)$ is greater than all elements from $C_{ij}(start)$ to $C_{ij}(start+q-1)$ and smaller than all elements from $C_{ij}(start+q+1)$ to $C_{ij}(end)$.
Therefore, we can select a single breakpoint from $C_{ij}$ by running \textit{QuickSelect} once. 
Since we set $N_r=256$, the \textit{divide-and-conquer} strategy can be perfectly applied to our algorithm. 
Specifically, a total of $\log_2 N_r$ rounds need to run, and the $z$-th round select $2^{z-1}$ breakpoints by running \textit{QuickSelect} in $2^{z-1}$ sub-regions generated from the ($z$-1)-th round, where $z=1,...,\log_2 N_r$.
For each $C_{ij}$, we select the minimum element as the first breakpoint $B_{ij}(1)$ and the maximum element as the last breakpoint $B_{ij}(N_r+1)$ (lines 9-11).
% Theoretically, the time complexity of the scheme that completely sorts $C_{ij}$ is $\mathcal{O}(K \cdot L \cdot n \cdot \log n)$. 
% Since the average time complexity of \textit{QuickSelect} is $\mathcal{O}(n)$, the time cost of $z$-th round is $2^{z-1} \cdot \mathcal{O}(\frac{n}{2^{z-1}})=\mathcal{O}(n)$. 
% so the total time complexity of Algorithm \ref{getbreakpoints} is $\mathcal{O}(K \cdot L \cdot n \cdot \log N_r)$.
In practice, Algorithm \ref{getbreakpoints} achieves 3x speedup in running time over the complete sorting scheme, as shown in Section \ref{selfevaluation}. 
After getting the breakpoint set $B$, Algorithm \ref{dynamic_encoding} will encode all points into iSAX representations and return the set of encoded points $EP$ (lines 3-8).
% Since we use \textit{BinarySearch} to locate the region where each data point belongs, the time complexity of Algorithm \ref{dynamic_encoding} is $\mathcal{O}(K \cdot L \cdot n \cdot \log N_r)$.

\begin{algorithm}[tb]
% \footnotesize

	\caption{Create Index}%算法名字                                                                        
	\label{create_index}
	\LinesNumbered %要求显示行号
	\KwIn{Parameters $K$, $L$, $n$, encoded points set $EP$, maximum size of a leaf node $max\_size$}%输入参数
	\KwOut{A set of DE-Trees: $DETs=[T_1,...,T_L]$}%输出
	\For{$i=1$ to $L$}{
		Initialize $T_i$ and generate $2^{K}$ first layer nodes as the original leaf nodes; \\
		\For{$z=1$ to $n$}{
			$ep_i(o_z) \leftarrow (EP_{i1}(o_z),...,EP_{iK}(o_z))$; \\
			$pos_z \leftarrow$ the position of $o_z$ in the dataset; \\
			$target\_leaf \leftarrow$ leaf node of $T_i$ to insert $\langle ep_i(o_z),pos_z \rangle$; \\
			\While{$sizeof(target\_leaf) \geq max\_size$}{
				SplitNode($target\_leaf$); \\
				$target\_leaf \leftarrow$ the new leaf node to insert $\langle ep_i(o_z),pos_z \rangle$; \\
			}
			Insert $\langle ep_i(o_z),pos_z \rangle$ to $target\_leaf$; \\
		}
	}
	\Return $DETs$; \\
\end{algorithm}

\subsection{Indexing Phase} \label{Indexing Phase}

As mentioned before, DET-LSH requires $L$ DE-Trees to support queries. 
Algorithm~\ref{create_index} presents how to construct $L$ DE-Trees based on the encoded points set $EP$. 
Specifically, for each DE-Tree, the first step is to initialize the first layer nodes, which are the children of the root (line 2). 
As shown in Figure \ref{isaxindex}, according to the iSAX encoding rules, the initial division of each dimension has two cases: $0^*$ and $1^*$. 
Therefore, each DE-Tree has $2^K$ first layer nodes. 
Then, for each point $o_z$, we get its encoded representation $ep_i(o_z)$ for the $i$-th DE-Tree $T_i$ and its position $pos_z$ in the dataset (lines 3-5). 
Based on $ep_i(o_z)$, we can get the leaf node of $T_i$ to insert $\langle ep_i(o_z),pos_z \rangle$ (line 6). 
If the leaf node is full, we split it until we get a new leaf node and insert $\langle ep_i(o_z),pos_z \rangle$ (lines 7-10). 
Note that only leaf nodes contain information about points, such as encoded representations and positions, while internal nodes only contain index information.

\begin{algorithm}[tb]
% \footnotesize

	\caption{DET Range Query}%算法名字                                                                      
	\label{range_query}
	\LinesNumbered %要求显示行号
	\KwIn{A projected query point $q^\prime$, the search radius $r^\prime$, the index DE-Tree $T$, project dimension $K$}%输入参数
	\KwOut{A set of points $S$}%输出
	Initialize a points set $S \leftarrow \varnothing$; \\
	\For{$i=1$ to $2^K$}{
		$node \leftarrow$ the $i$-th child of root node in $T$; \\
		\textbf{call} TraverseSubtree($node$, $q^\prime$, $r^\prime$, $S$); \\
	}
	\Return $S$;
\end{algorithm}

\begin{algorithm}[tb]
% \footnotesize

	\caption{Traverse Subtree}%算法名字                                                                      
	\label{traversesubtree}
	\LinesNumbered %要求显示行号
	\KwIn{A node $node$, the projected query point $q^\prime$, the search radius $r^\prime$, the set of points $S$}%输入参数
	$lower\_bound\_dist \leftarrow$ calculate the lower bound distance between $q^\prime$ and $node$; \\
	\If{$lower\_bound\_dist > r^\prime$}{
		break; \\
	} \ElseIf{$node$ is a leaf} {
		$upper\_bound\_dist \leftarrow$ calculate the upper bound distance between $q^\prime$ and $node$; \\
		\If{$upper\_bound\_dist \leq r^\prime$}{
			$S \leftarrow S \, \cup$ all points in $node$;  \\
		} \Else {
			\While{$node$ has next point}{
				Get next point $o \in node$; \\
				$dist \leftarrow$ calculate the distance between $q^\prime$ and the projected $o^\prime$; \\
				\If{$dist \leq r^\prime$}{
					$S \leftarrow S \, \cup o$; \\
				}
			}
		}
	} \Else {
		\textbf{call} TraverseSubtree($node.leftChild$, $q^\prime$, $r^\prime$, $S$); \\
		\textbf{call} TraverseSubtree($node.rightChild$, $q^\prime$, $r^\prime$, $S$); \\
	}
\end{algorithm}

In a DE-Tree, except the root node that has $2^K$ children, other internal nodes have only two children. 
This is because when an internal node needs to be split, we only select one of its $K$ dimensions for further bit-wise binary division. 
For example, in Figure~\ref{isaxindex}, we choose the first dimension of node [$0^*$,$0^*$] to split, 
and the representations of its two children are [$00$,$0^*$] and [$01$,$0^*$]. 
The choice of which dimension to divide is important for splitting nodes. 
Intuitively, splitting a node works better if the obtained two children contain similar numbers of points. 
Therefore, when splitting nodes, we choose the dimension that most evenly divides the points.

\subsection{Query Phase} \label{Query Phase}
 
Since the query strategy of DET-LSH is based on the Euclidean distance metric, range queries can improve the efficiency of obtaining candidate points. 
In a DE-Tree, each space is divided into different regions by multiple breakpoints. 
The breakpoints on all sides of a region can be used to calculate the upper and lower bound distances between two points or between a point and a tree node. 

\noindent \textbf{DET Range Query.} Algorithm~\ref{range_query} is designed for range queries in DE-Tree. 
% Intuitively, a range query needs to traverse as few nodes of the DE-Tree as possible and find the leaf nodes containing data points that are within a distance range. 
We select all $2^K$ children of the root node as the entry of the traversal and then traverse their subtrees in order (lines 2-4). 
Algorithm~\ref{traversesubtree} presents how to obtain points within the search radius $r^\prime$ by recursively traversing the subtrees. 
For the node being visited, if its lower bound distance with $q^\prime$ is greater than $r^\prime$, it means that the distance between any point in its subtree and $q^\prime$ is greater than $r^\prime$, so no further traversal is needed (lines 1-3). 
If the upper bound distance between a leaf node and $q^\prime$ is not greater than $r^\prime$, it means that the distance between any data point in the leaf node and $q^\prime$ is not greater than $r^\prime$, so that all points can be added to $S$ (lines 4-7). 
If $r^\prime$ falls within the range of the lower bound distance and the upper bound distance between a leaf node and $q^\prime$, we should traverse the data points in the leaf node and add those within the search radius to $S$ (lines 8-13). 
If the node being visited is not a leaf node and further traversal is required, we need to further traverse its subtrees (lines 14-16).
% The maximum height of a DE-Tree is $K \cdot \log_2 N_r$, and each DE-Tree has $O(\frac{n}{max\_size})$ leaf nodes. The time complexity of Algorithm~\ref{traversesubtree} in $L$ DE-Trees consists of two parts: (1) calculating the distance between the query point and all traversed tree nodes (lines 1 and 5) has $O(L \cdot K^2 \cdot \log N_r \cdot \frac{n}{max\_size})$ cost; and (2) calculating the distance between the query point and all points in some nodes (lines 8-13) has $O(L \cdot K \cdot \frac{n} {max\_size} \cdot max\_size) = O(L \cdot K \cdot n)$ cost.

\begin{algorithm}[tb]
	\caption{($r$,$c$)-ANN Query}%算法名字   
 % \footnotesize
 
	\label{rcann}
	\LinesNumbered %要求显示行号
	\KwIn{A query point $q$, parameters $K$, $L$, $n$, $c$, $r$, $\epsilon$, $\beta$, index DE-Trees $DETs=[T_1,...,T_L]$}%输入参数
	\KwOut{A point $o$ or $\varnothing$}%输出
	Initialize a candidate set $S \leftarrow \varnothing$; \\
	\For{$i=1$ to $L$}{
		Compute $q_i^\prime=H _i(q)=[h_{i1}(q),...,h_{iK}(q)]$; \\
		$S_i \leftarrow$ \textbf{call} DETRangeQuery($q_i^\prime,\epsilon \cdot r, T_i, K$); \\
		$S \leftarrow S \cup S_i$; \\
		\If{$\lvert S \rvert \geq \beta n+1$}{
			\Return the point $o$ closest to $q$ in $S$; \\
		}
	}
	\If{$\lvert \left\{ o \mid o \in S \land \left\|o,q\right\| \leq c \cdot r \right\} \rvert \geq 1$}{
		\Return the point $o$ closest to $q$ in $S$; \\
	}
	\Return $\varnothing$;
\end{algorithm}

\noindent \textbf{($r$,$c$)-ANN Query.} Algorithm \ref{rcann} shows that DET-LSH can answer an ($r$,$c$)-ANN query with any search radius $r$. 
After the indexing phase, DET-LSH obtains $L$ DE-Trees $T_1,...,T_L$. Given a query $q$, we consider $L$ projected spaces in order. 
For the $i$-th space, we first compute the projected query $q_i^\prime$ (line 3).
% which has $O(L \cdot K \cdot d)$ time cost. 
Then, we call Algorithm \ref{range_query} to perform a range query in the $i$-th DE-Tree $T_i$ (line 4). 
The search radius in the projected space is $\epsilon \cdot r$. The parameter $\epsilon$ guarantees that if the distance between a point $o$ and $q$ is not greater than $r$, 
then the distance between the projected $o^\prime$ and $q^\prime$ is not greater than $\epsilon \cdot r$ with a constant probability. 
Detailed analysis and proof will be introduced in Lemma \ref{lemma3} in Section \ref{chapter5}.
We continuously add the candidate points obtained by range queries to a candidate set $S$ (line 5). 
If the number of candidate points in $S$ exceeds $\beta n+1$, the point $o$ closest to $q$ will be returned, 
where parameter $\beta$ is the maximum false positive percentage (lines 6-7). 
After completing range queries in $L$ DE-Trees, if the size of $S$ is still smaller than $\beta n+1$ and there is at least one point in $S$ whose distance with $q$ is not greater than $c \cdot r$, then return the point $o$ closest to $q$ in $S$ (lines 8-9). 
Otherwise, the algorithm returns nothing (line 10).
% Computing the real distance of each candidate point to $q$ has $\mathcal{O}(\beta nd)$ time cost and finding the result has $\mathcal{O}(\log (\beta n))$ time cost.
According to Theorem \ref{theorem1}, to be introduced in Section \ref{chapter5}, DET-LSH can correctly answer an ($r$,$c$)-ANN query with a constant probability.

\noindent \textbf{$c^2$-$k$-ANN Query.} Since $o^*$ and $\left\|q,o^*\right\|$ are not known in advance, 
we cannot directly perform an ANN query with a pre-defined $r$ like ($r$,$c$)-ANN query does. 
Instead, we can conduct a series of ($r$,$c$)-ANN queries with increasing radii until enough points are returned. 
Algorithm \ref{ckann} outlines the query processing. 
We can see that most of the steps of Algorithm \ref{ckann} (lines 3-10) are almost the same as Algorithm \ref{rcann} (lines 2-9), except that Algorithm \ref{ckann} needs to consider $k$ when judging conditions and returning results. 
The main difference is that when neither the termination condition at line 8 nor line 10 is satisfied, 
Algorithm \ref{ckann} will enlarge the search radius for the next round of queries (line 11). 
According to Theorem \ref{theorem2}, to be introduced in Section \ref{chapter5}, DET-LSH can correctly answer a $c^2$-$k$-ANN query with a constant probability.

\begin{algorithm}[tb]
	\caption{$c^2$-$k$-ANN Query}%算法名字                
 % \footnotesize

	\label{ckann}
	\LinesNumbered %要求显示行号
	\KwIn{A query point $q$, parameters $K$, $L$, $n$, $c$, $r_{min}$, $\epsilon$, $\beta$, $k$, index DE-Trees $DETs=[T_1,...,T_L]$}%输入参数
	\KwOut{$k$ nearest points to $q$ in $S$}%输出
	Initialize a candidate set $S \leftarrow \varnothing$ and set $r \leftarrow r_{min}$; \\
	\While{\textit{TRUE}}{
		\For{$i=1$ to $L$}{
			Compute $q_i^\prime=H _i(q)=[h_{i1}(q),...,h_{iK}(q)]$; \\
			$S_i \leftarrow$ \textbf{call} DETRangeQuery($q_i^\prime,\epsilon \cdot r, T_i, K$); \\
			$S \leftarrow S \cup S_i$; \\
			\If{$\lvert S \rvert \geq \beta n+k$}{
				\Return the \textit{top}-$k$ points closest to $q$ in $S$; \\
			}
		}
		\If{$\lvert \left\{ o \mid o \in S \land \left\|o,q\right\| \leq c \cdot r \right\} \rvert \geq k$}{
			\Return the \textit{top}-$k$ points closest to $q$ in $S$; \\
		}
		$r \leftarrow c \cdot r$;
	}
\end{algorithm}

\section{Theoretical Analysis} \label{chapter5}

\subsection{Quality Guarantee}  \label{chapter5.1}

Let $\mathcal H _i(o)=[h_{i1}(o),...,h_{iK}(o)]$ denote a data point $o$ in the $i$-th projected space, where $i=1,...,L$. We define three events as follows:

\begin{itemize}
	\item \textbf{E1:} If there exists a point $o$ satisfying $\left\|o,q\right\| \leq r$, then its projected distance to $q$, i.e., $\left\|\mathcal H _i(o),\mathcal H _i(q)\right\|$, is smaller than $\epsilon r$ for some $i=1,...,L$;
	\item \textbf{E2:} If there exists a point $o$ satisfying $\left\|o,q\right\| > cr$, then its projected distance to $q$, i.e., $\left\|\mathcal H _i(o),\mathcal H _i(q)\right\|$, is smaller than $\epsilon r$ for some $i=1,...,L$;
	\item \textbf{E3:} Fewer than $\beta n$ points satisfying \textbf{E2} in dataset $\mathcal D$.
\end{itemize}

\begin{lemma}\label{lemma3}
	Given $K$ and $c$, setting $L=-\frac{1}{\ln{\alpha_1}}$ and $\beta=2-2\alpha_2^{-\frac{1}{ln\alpha_1}}$ such that $\alpha_1$, $\alpha_2$, and $\epsilon$ satisfy Equation \ref{eq2}, the probability that \textbf{E1} occurs is at least $1-\frac{1}{\mathrm{e}}$ and the probability that \textbf{E3} occurs is at least $\frac{1}{2}$.
	
	\begin{equation}  \label{eq2}
%		\epsilon^2=\chi^2_{\alpha_1}(K)=c^2 \chi^2_{1-\alpha_2}(K).
		\epsilon^2=\chi^2_{\alpha_1}(K)=c^2 \cdot \chi^2_{\alpha_2}(K).
	\end{equation}
\end{lemma}

\begin{proof}
	Given a point $o$ satisfying $\left\|o,q\right\| \leq r$, let $s=\left\|o,q\right\|$ and $s_i^\prime=\left\|\mathcal H _i(o),\mathcal H _i(q)\right\|$ denote the distances between $o$ and $q$ in the original space and in the $i$-th projected space, where $i=1,...,L$. From Equation \ref{eq2}, we have $\sqrt{\chi^2_{\alpha_1}(K)}=\epsilon$. For each independent projected space, from Lemma \ref{lemma2}, we have $\Pr{[s_i^\prime>s\sqrt{\chi^2_{\alpha_1}(K)}]} = \Pr{[s_i^\prime>\epsilon s]} = \alpha_1$. Since $s \leq r$, $\Pr{[s_i^\prime>\epsilon r]} \leq \alpha_1$. Considering $L$ projected spaces, we have $\Pr{[\textbf{E1}]}\geq1-\alpha_1^L=1-\frac{1}{\mathrm{e}}$. Likewise, given a point $o$ satisfying $\left\|o,q\right\| > cr$, let $s=\left\|o,q\right\|$ and $s_i^\prime=\left\|\mathcal H _i(o),\mathcal H _i(q)\right\|$ denote the distances between $o$ and $q$ in the original space and in the $i$-th projected space, where $i=1,...,L$. From Equation \ref{eq2}, we have $\sqrt{\chi^2_{\alpha_2}(K)}=\frac{\epsilon}{c}$. For each independent projected space, from Lemma \ref{lemma2}, we have $\Pr{[s_i^\prime>s\sqrt{\chi^2_{\alpha_2}(K)}]} = \Pr{[s_i^\prime>\frac{\epsilon s}{c}]} = \alpha_2$. Since $s > cr$, i.e., $\frac{s}{c} > r$, $\Pr{[s_i^\prime>\epsilon r]} > \alpha_2$. Considering $L$ projected spaces, we have $\Pr{[\textbf{E2}]} \leq 1-\alpha_2^L$, thus the expected number of such points in dataset $\mathcal D$ is upper bounded by $(1-\alpha_2^L) \cdot n$. By \textit{Markov's inequality}, we have $\Pr{[\textbf{E3}]}>1-\frac{(1-\alpha_2^L) \cdot n}{\beta n}=\frac{1}{2}.$
	
%	When $L=1$, by \textit{Markov's inequality}, we have $\Pr{[\textbf{E2} \mid L=1]}>1-\frac{\alpha_2}{\beta}$. Since $L$ projected spaces are generated independently, if all projected spaces satisfy $\Pr{[\textbf{E2} \mid L=1]}$, the \textbf{E2} event must hold. Therefore, we have $\Pr{[\textbf{E2}]} > \Pr{[\textbf{E2} \mid L=1]}>1-\frac{\alpha_2}{\beta}$
\end{proof}

\begin{theorem}\label{theorem1}
	Algorithm \ref{rcann} answers an ($r$,$c$)-ANN query with at least a constant probability of $\frac{1}{2}-\frac{1}{\mathrm{e}}$.
\end{theorem}

\begin{proof}
	We show that when $\textbf{E1}$ and $\textbf{E3}$ hold at the same time, 
 Algorithm \ref{rcann} returns an correct ($r$,$c$)-ANN result. 
 The probability of \textbf{E1} and \textbf{E3} occurring at the same time can be calculated as $\Pr{[\textbf{E1}\textbf{E3}]}=\Pr{[\textbf{E1}]}-\Pr{[\textbf{E1}\overline{\textbf{E3}}]} > \Pr{[\textbf{E1}]}-\Pr{[\overline{\textbf{E3}}]}=\frac{1}{2}-\frac{1}{\mathrm{e}}$. 
 When $\textbf{E1}$ and $\textbf{E3}$ hold at the same time, if Algorithm \ref{rcann} terminates after getting at least $\beta n+1$ candidate points (line 7),
 due to $\textbf{E3}$, there are at most $\beta n$ points satisfying $\left\|o,q\right\| > cr$. 
 Thus we can get at least one point satisfying $\left\|o,q\right\| \leq cr$, and the returned point is obviously a correct result. 
 If the candidate set $S$ has no more than $\beta n+1$ points, but there exists at least one point in $S$ satisfying $\left\|o,q\right\| \leq cr$, 
 Algorithm \ref{rcann} can also terminate and then return a result correctly (line 9). 
 Otherwise, it indicates that no points satisfying $\left\|o,q\right\| \leq cr$. 
 According to the Definition \ref{def3} of ($r$,$c$)-ANN, nothing will be returned (line 10). 
 Therefore, when $\textbf{E1}$ and $\textbf{E3}$ hold at the same time, Algorithm \ref{rcann} can always correctly answer an ($r$,$c$)-ANN query. 
 In other words, Algorithm \ref{rcann} answers an ($r$,$c$)-ANN query with at least a constant probability of $\frac{1}{2}-\frac{1}{\mathrm{e}}$.
\end{proof}

\begin{theorem}\label{theorem2}
	Algorithm \ref{ckann} answers a $c^2$-$k$-ANN query with at least a constant probability of $\frac{1}{2}-\frac{1}{\mathrm{e}}$.
    % Algorithm \ref{ckann} answers an ANN query with an approximation ratio $c^2$, i.e., $c^2$-$k$-ANN, with at least a constant probability of $\frac{1}{2}-\frac{1}{\mathrm{e}}$.
\end{theorem}

\begin{proof}
	We show that when $\textbf{E1}$ and $\textbf{E3}$ hold at the same time, Algorithm \ref{ckann} returns a correct $c^2$-$k$-ANN result. 
 Let $o_i^*$ be the $i$-th exact nearest point to $q$ in $\mathcal D$, we assume that $r_i^* = \left\|o_i^*,q\right\| > r_{min}$, where $r_{min}$ is the initial search radius and $i=1,...,k$. 
 We denote the number of points in the candidate set under search radius $r$ as $\lvert S_r \rvert$. 
 Obviously, when enlarging the search radius $r=r_{min}, r_{min} \cdot c, r_{min} \cdot c^2,...$, there must exist a radius $r_0$ satisfying $\lvert S_{r_0} \rvert < \beta n+k$ and $\lvert S_{c \cdot r_0} \rvert \geq \beta n+k$. 
 The distribution of $r_i^*$ has three cases:
	\begin{enumerate}
		\item \textbf{Case 1:} If for all $i=1,...,k$ satisfying $r_i^* \leq r_0$, which indicates the range queries in all $L$ index trees have been executed at $r=r_0$ (lines 3-8). Due to $\textbf{E1}$, all $r_i^*$ must in $S_{r_0}$. Since $S_{r_0} \subsetneqq S_{c \cdot r_0}$, all $r_i^*$ also must in $S_{c \cdot r_0}$. Therefore, Algorithm \ref{ckann} returns the exact $k$ nearest points $o_i^*$ to $q$.
		\item \textbf{Case 2:} If for all $i=1,...,k$ satisfying $r_i^* > r_0$, all $r_i^*$ not belong to $S_{r_0}$. Since Algorithm \ref{ckann} may terminate after executing range queries in part of $L$ index trees at $r=c \cdot r_0$ (line 8), we cannot guarantee that $r_i^* \leq c \cdot r_0$. However, due to $\textbf{E3}$, there are at least $k$ points $o_i$ in $S_{c \cdot r_0}$ satisfying $\left\|o_i,q\right\| \leq c^2r_0$, $i=1,...,k$. Therefore, we have $\left\|o_i,q\right\| \leq c^2r_0 \leq c^2r_i^*$, i.e., each $o_i$ is a $c^2$-ANN point for corresponding $o_i^*$.
		\item \textbf{Case 3:} If there exists an integer $m \in (1,k)$  such that for all $i=1,...,m$ satisfying $r_i^* \leq r_0$ and for all $i=m+1,...,k$ satisfying $r_i^* > r_0$, indicating that \textbf{Case 3} is a combination of \textbf{Case 1} and \textbf{Case 2}. For each $i \in [1,m]$, Algorithm \ref{ckann} returns the exact nearest point $o_i^*$ to $q$ based on \textbf{Case 1}. For each $i \in [m+1,k]$, Algorithm \ref{ckann} returns a $c^2$-ANN point for $o_i^*$ based on \textbf{Case 2}.
	\end{enumerate}
	Therefore, when $\textbf{E1}$ and $\textbf{E3}$ hold simultaneously, Algorithm \ref{ckann} can always correctly answer a $c^2$-$k$-ANN query, i.e., Algorithm \ref{ckann} returns a $c^2$-$k$-ANN with at least a constant probability of $\frac{1}{2}-\frac{1}{\mathrm{e}}$.
\end{proof}

\subsection{Parameter Settings} \label{chapter5.2}

\begin{figure}[tb] 
	\centering
	\includegraphics[width=0.59\linewidth]{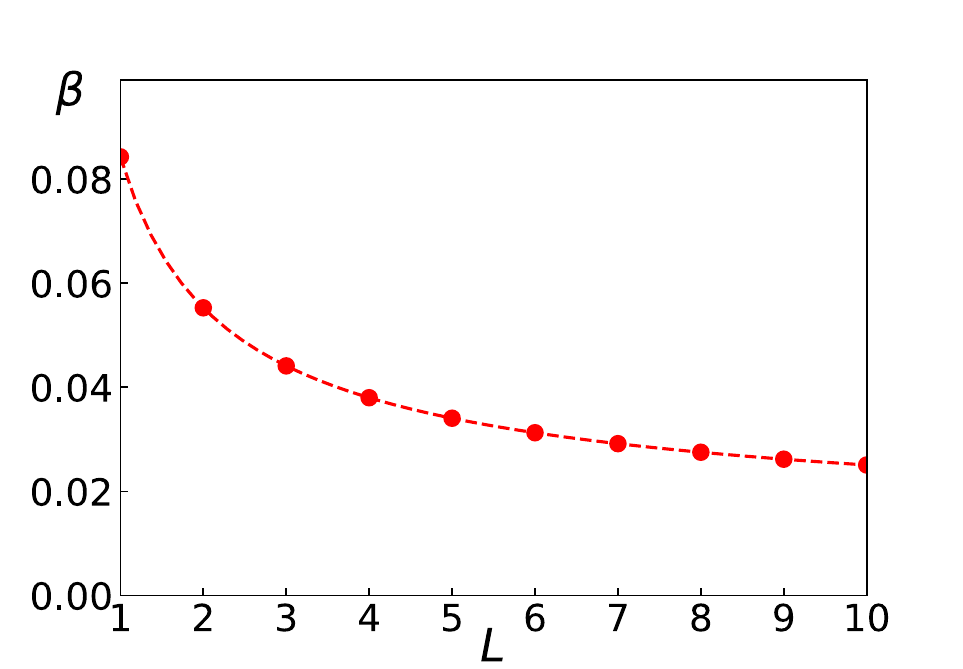}
	\caption{Illustration of the theoretical $\beta$ when $L$ varies (for $K=16$ and $c=1.5$), which is in line with Lemma \ref{lemma3}.}
	\label{betal}
\end{figure}

The performance of DET-LSH is affected by several parameters: $L$, $K$, $\beta$, $c$, and so on. 
According to Lemma \ref{lemma3}, when $L$ and $c$ are set as constants, there is a mathematical relationship between $K$ and $\beta$. 
We set $K=16$ and $c=1.5$ by default, and Figure \ref{betal} shows the theoretical $\beta$ as $L$ changes, which is in line with Lemma \ref{lemma3}. Figure \ref{betal} illustrates that $\beta$ and $L$ have a negative correlation. 
Theoretically, a greater $\beta$ means a higher fault tolerance when querying, so the accuracy of DET-LSH is improved. 
Meanwhile, a greater $L$ means fewer correct results are missed when querying, so the accuracy of DET-LSH can also be improved. 
However, both greater $\beta$ and greater $L$ will reduce query efficiency, so we need to find a balance between $\beta$ and $L$. 
As shown in Figure \ref{betal}, $L=4$ is a good choice because as $L$ increases, $\beta$ drops rapidly until $L=4$, and then $\beta$ drops slowly. 
Therefore, we choose $L=4$ as the default value. 

For the initial search radius $r_{min}$, we follow the selection scheme proposed in \cite{pmlsh}. Specifically, to reduce the number of iterations for different $r$ and terminate the query process faster, 
we find a \enquote{magic} $r_{min}$ that satisfies the following conditions: 1) when $r=r_{min}$ in Algorithm \ref{ckann}, the number of candidate points in $S$ satisfies $\lvert S \rvert \geq \beta n+k$; 
2) when $r=\frac{r_{min}}{c}$ in Algorithm \ref{ckann}, the number of candidate points in $S$ satisfies $\lvert S \rvert < \beta n+k$. 
Since DET-LSH can implement dynamic incremental queries as $r$ increases, the choice of $r_{min}$ is expected to have a relatively small impact on its performance.

\subsection{Complexity Analysis}  \label{chapter5.3}

In the encoding and indexing phases, DET-LSH has time cost $\mathcal{O}(n(d+\log N_r))$, and space cost $\mathcal{O}(n)$.
The time cost comes from four parts: (1) computing hash values for $n$ points, $\mathcal{O}(L \cdot K \cdot n \cdot d)$; 
(2) using Algorithm \ref{getbreakpoints} for breakpoint selection, $\mathcal{O}(L \cdot K \cdot n \cdot \log N_r)$; 
(3) using Algorithm \ref{dynamic_encoding} for encoding, $\mathcal{O}(L \cdot K \cdot n \cdot \log N_r)$; 
and (4) using Algorithm \ref{create_index} for constructing $L$ DE-Trees, $\mathcal{O}(L \cdot n \cdot K \cdot \log N_r)$. 
Since both $K=\mathcal{O}(1)$ and $L=\mathcal{O}(1)$ are constants, the total time cost is $\mathcal{O}(n(d+\log N_r))$. 
Obviously, the size of encoded points and $L$ DE-Trees are both $\mathcal{O}(L \cdot K \cdot n)=\mathcal{O}(n)$.

In the query phase, DET-LSH has time cost $\mathcal{O}(n(\beta d+\log N_r))$. 
The time cost comes from four parts: 
(1) computing hash values for the query point $q$, $\mathcal{O}(L \cdot K \cdot d)=\mathcal{O}(d)$; 
(2) finding candidate points in $L$ DE-Trees, $\mathcal{O}(L \cdot K^2 \cdot \log N_r \cdot \frac{n}{max\_size} + L \cdot K \cdot n)=\mathcal{O}(n\log N_r)$; 
(3) computing the real distance of each candidate point to $q$, $\mathcal{O}(\beta nd)$; 
and (4) finding the $top$-$k$ points to $q$, $\mathcal{O}(\beta n \log k)$. 
The total time cost in the query phase is $\mathcal{O}(n(\beta d+\log N_r))$.

\begin{table}
	\centering
	\caption{Datasets}
	\label{table2}
        \resizebox{\linewidth}{!}{
	\begin{tabular}{cccc}
		\toprule
		\textbf{Dataset} & \textbf{Cardinality} & \textbf{Dimensions} & \textbf{Type} \\
		\midrule
		Msong & 994,185 & 420 & Audio\\
		Deep1M & 1,000,000 & 256 & Image\\
		Sift10M & 10,000,000 & 128 & Image\\
		TinyImages80M & 79.302,017 & 384 & Image\\
		Sift100M & 100,000,000 & 128 & Image\\
		Yandex Deep500M & 500,000,000 & 96 & Image\\
		Microsoft SPACEV500M & 500,000,000 & 100 & Text\\
		Microsoft Turing-ANNS500M & 500,000,000 & 100 & Text\\
		\bottomrule
	\end{tabular}}
\end{table}

\section{Experimental Evaluation} \label{chapter6}

In this section, we self-evaluate DET-LSH, conduct comparative experiments with the state-of-the-art LSH-based methods, and compare with graph-based methods.
Our method is implemented in C and C++ and compiled using -O3 optimization. 
All experiments are conducted using a single thread, on a machine with 2 AMD EPYC 9554 CPUs @ 3.10GHz and 756 GB RAM, running on Ubuntu 22.04.

\subsection{Experimental Setup}

\noindent \textbf{Datasets and Queries.} We use eight real-world datasets for ANN search. 
Table \ref{table2} shows the key statistics of the datasets. Note that the points in \textit{Sift10M} and \textit{Sift100M} are randomly chosen from the \textit{Sift1B} dataset\footnote{http://corpus-texmex.irisa.fr/}. 
Similarly, the points in \textit{Yandex Deep500M}, \textit{Microsoft SPACEV500M}, and \textit{Microsoft Turing-ANNS500M} are also randomly chosen from their 1B-scale datasets\footnote{https://big-ann-benchmarks.com/neurips21.html}. 
We randomly select 100 data points as queries and remove them from the original datasets.

\noindent \textbf{Evaluation Measures.} We adopt five measures to evaluate the performance of all methods: index size, indexing time, query time, recall, and overall ratio.
% We adopt four measures to evaluate the performance of all methods: indexing time, query time, recall, and overall ratio, 
% where the indexing time and query time evaluate the efficiency of methods, and the recall and overall ratio evaluate the quality of results. 
For a query $q$, we denote the result set as $R=\{o_1,...,o_k\}$ and the exact $k$-NNs as $R^*=\{o_1^*,...,o_k^*\}$, recall is defined as $\frac{\lvert R \cap R^* \rvert}{k}$ and overall ratio is defined as $\frac{1}{k} \sum_{i=1}^{k} \frac{\left\|q,o_i\right\|}{\left\|q,o_i^*\right\|}$ \cite{dblsh}.

%\begin{equation}  \label{recall}
%	Recall = \frac{\lvert R \cap R^* \rvert}{k}
%\end{equation}
%
%\begin{equation}  \label{overallratio}
%	Overall Ratio = \frac{1}{k} \sum_{i=1}^{k} \frac{\left\|q,o_i\right\|}{\left\|q,o_i^*\right\|}
%\end{equation}

\noindent \textbf{Benchmark Methods.} We compare DET-LSH with three state-of-the-art LSH-based in-memory methods mentioned in Section \ref{chapter2}, i.e., DB-LSH \cite{dblsh}, LCCS-LSH \cite{lccslsh}, and PM-LSH \cite{pmlsh}. 
% As mentioned in Section \ref{chapter2}, existing mainstream LSH-based methods can be classified into three categories: BC, C2, and DM. 
% DB-LSH, LCCS-LSH, and PM-LSH are the state-of-the-art methods of BC, C2, and DM, respectively. 
Moreover, to study the capability of DE-Tree and the advantages of LSH, we use a single DE-Tree to index points without LSH for ANN searches. 
We call this method DET-ONLY. 
Since DET-ONLY is not based on LSH, we adopt the Piecewise Aggregate Approximation (PAA) \cite{keogh2001} technique to reduce the dimensionality of points. 
PAA divides a $d$-dimensional point into $K$ segments of equal length $\lfloor \frac{d}{K} \rfloor$ and 
uses the mean value of the coordinates in each segment to summarize the point. 
DET-ONLY adopts the same query strategy as DET-LSH.
To study the characteristics of LSH-based methods and graph-based methods, we also compare DET-LSH with two state-of-the-art graph-based methods, i.e., HNSW \cite{malkov2018efficient} and LSH-APG \cite{lshapg}.

\noindent \textbf{Parameter Settings.} $k$ in $k$-ANN is set to 50 by default.
For DET-LSH, the parameters are set as described in Section \ref{chapter5.2}. 
For competitors, the parameter settings follow their source codes or papers. 
To make a fair comparison, we set $\beta=0.1$ and $c=1.5$ for DET-LSH, DB-LSH, PM-LSH, and DET-ONLY. 
For DB-LSH, $L=5$, $K=12$, $w=4c^2$. 
For LCCS-LSH, $m=64$. For PM-LSH, $s=5$, $m=15$. For DET-ONLY, $K=16$, $L=1$.
For HNSW, $M=48$, $ef=100$. For LSH-APG, $K=16$, $L=2$, $T=24$, $T^{\prime}=2T$, $p_{\tau}=0.95$.

\begin{figure}[tb] 
	\centering
	\includegraphics[width=\linewidth]{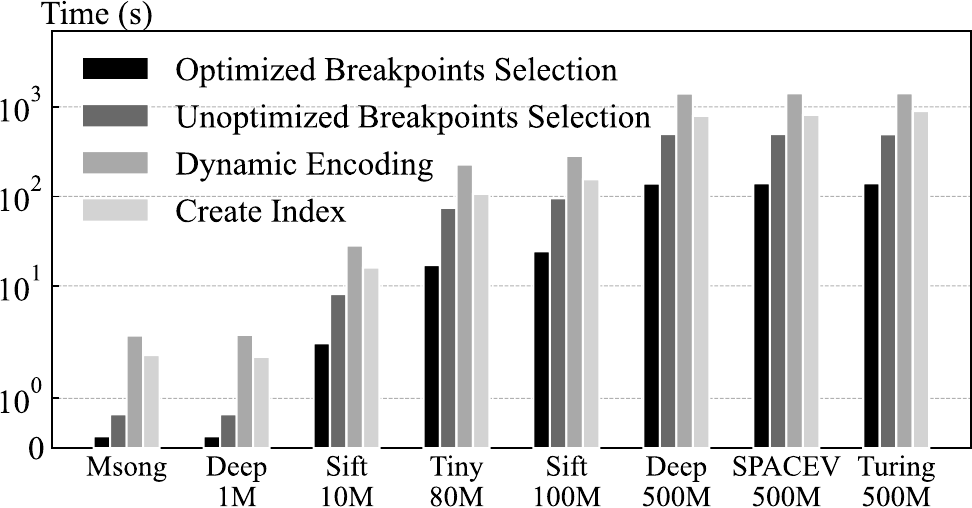}
	\caption{Running time break-down for the DET-LSH encoding and indexing phases.}
	\label{encodingandindexing}
\end{figure}

\begin{figure}[tb] 
	\centering
	\includegraphics[width=\linewidth]{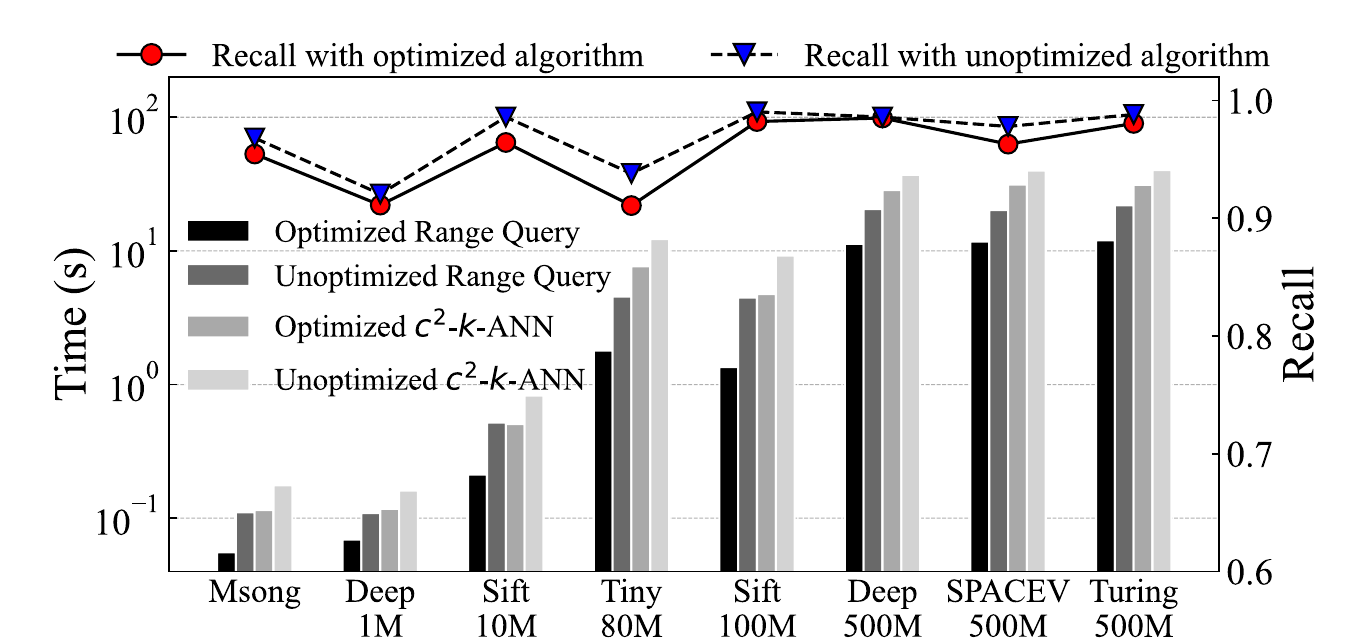}
	\caption{Running time and recall of optimized/non-optimized query-phase algorithms of DET-LSH.}
	\label{querytime}
\end{figure}

\begin{figure}[tb] 
	\centering
	\includegraphics[width=\linewidth]{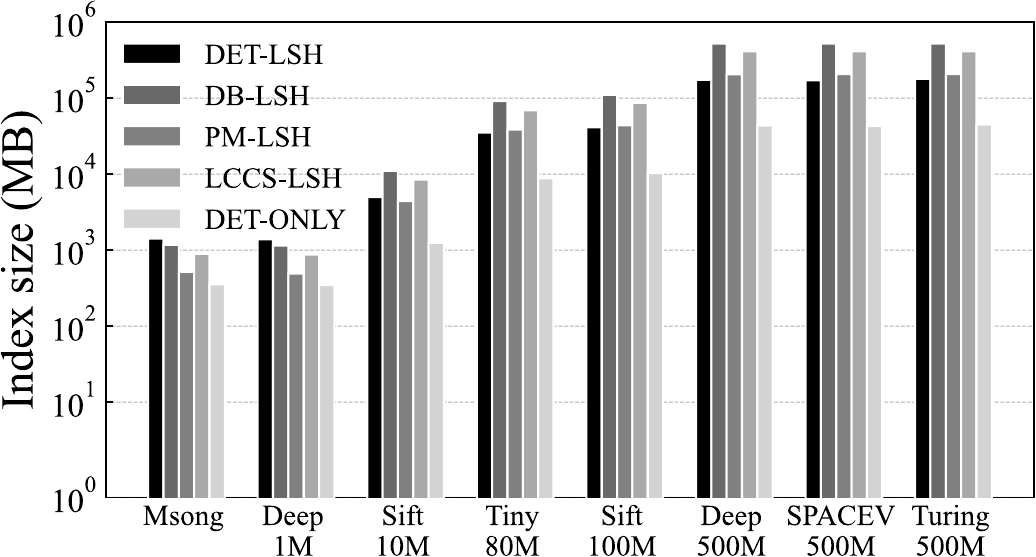}
	\caption{Index size for all datasets.}
	\label{indexsize}
\end{figure}

\begin{table*}[]
	\centering
	\caption{Performance comparison with competitors (the best value in each row is highlighted in bold; the number in parentheses indicates how many times slower a method is than the best method).}
	\label{table3}
\begin{threeparttable}
{\small
	\begin{tabular}{|cc|m{1.9cm}<{\centering}|m{1.9cm}<{\centering}|m{1.9cm}<{\centering}|m{1.9cm}<{\centering}|m{1.9cm}<{\centering}|}
		\hline
\multicolumn{2}{|c|}{}                                                                        & \textbf{DET-LSH} & \textbf{DB-LSH} & \textbf{PM-LSH} & \textbf{LCCS-LSH} & \textbf{DET-ONLY} \\ \hline
\multicolumn{1}{|c|}{\multirow{4}{*}{\textbf{Msong}}}                     & Query Time (ms)   & 112.97 (1.43)  & 118.10 (1.49)          & 120.36 (1.53)          & 170.13 (2.16)            & \textbf{78.87}             \\ \cline{2-7} 
\multicolumn{1}{|c|}{}                                                    & Recall            & \textbf{0.9546}  & 0.9474          & 0.949           & 0.849             & 0.891             \\ \cline{2-7} 
\multicolumn{1}{|c|}{}                                                    & Overall Ratio     & \textbf{1.0012}  & 1.0013          & 1.0013          & 1.0035            & 1.0046            \\ \cline{2-7} 
\multicolumn{1}{|c|}{}                                                    & Indexing Time (s) & 4.654 (3.90)   & 4.974 (4.17)          & 2.950 (2.47)           & 28.925 (24.2)            & \textbf{1.194}             \\ \hline
\multicolumn{1}{|c|}{\multirow{4}{*}{\textbf{Deep1M}}}                    & Query Time (ms)   & 109.28 (1.37)  & 117.79 (1.48)          & 207.37 (2.61)         & 136.21 (1.71)           & \textbf{79.51}             \\ \cline{2-7} 
\multicolumn{1}{|c|}{}                                                    & Recall            & \textbf{0.9112}  & 0.8552          & 0.857           & 0.848             & 0.818             \\ \cline{2-7} 
\multicolumn{1}{|c|}{}                                                    & Overall Ratio     & \textbf{1.0022}  & 1.0038          & 1.0042          & 1.0039            & 1.0061            \\ \cline{2-7} 
\multicolumn{1}{|c|}{}                                                    & Indexing Time (s) & 4.647 (3.97)   & 4.809 (4.10)           & 2.991 (2.55)          & 57.652 (49.2)           & \textbf{1.172}             \\ \hline
\multicolumn{1}{|c|}{\multirow{4}{*}{\textbf{Sift10M}}}                   & Query Time (ms)   & 506.34 (1.20)  & 944.23 (2.25)         & 1482.84 (3.53)        & 1905.09 (4.53)          & \textbf{420.43}            \\ \cline{2-7} 
\multicolumn{1}{|c|}{}                                                    & Recall            & \textbf{0.9644}  & 0.9438          & 0.9338          & 0.8924            & 0.886             \\ \cline{2-7} 
\multicolumn{1}{|c|}{}                                                    & Overall Ratio     & \textbf{1.0009}  & 1.0015          & 1.0016          & 1.0021            & 1.0035            \\ \cline{2-7} 
\multicolumn{1}{|c|}{}                                                    & Indexing Time (s) & 44.435 (4.00)  & 64.861 (5.85)         & 80.099 (7.22)         & 509.417 (45.9)          & \textbf{11.094}            \\ \hline
\multicolumn{1}{|c|}{\multirow{4}{*}{\textbf{TinyImages80M}}}             & Query Time (ms)   & 7676.23 (1.00) & 8164.96 (1.07)        & 13672.6 (1.79)        & 11272.8 (1.47)          & \textbf{7657.08}           \\ \cline{2-7} 
\multicolumn{1}{|c|}{}                                                    & Recall            & \textbf{0.9108}  & 0.9056          & 0.8822          & 0.87              & 0.8338            \\ \cline{2-7} 
\multicolumn{1}{|c|}{}                                                    & Overall Ratio     & \textbf{1.0016}  & \textbf{1.0016}          & 1.0023          & 1.0019            & 1.0036            \\ \cline{2-7} 
\multicolumn{1}{|c|}{}                                                    & Indexing Time (s) & 335.419 (4.08) & 641.988 (7.81)        & 1471.31 (17.89)        & 12128.1 (147.5)          & \textbf{82.235}            \\ \hline
\multicolumn{1}{|c|}{\multirow{4}{*}{\textbf{Sift100M}}}                  & Query Time (ms)   & 4757.76 (1.20) & 11064.8 (2.78)        & 15722.8 (3.95)        & 24221.8 (6.08)          & \textbf{3983.41}           \\ \cline{2-7} 
\multicolumn{1}{|c|}{}                                                    & Recall            & \textbf{0.9822}  & 0.9652          & 0.944           & 0.892             & 0.8848            \\ \cline{2-7} 
\multicolumn{1}{|c|}{}                                                    & Overall Ratio     & \textbf{1.0005}  & 1.0007          & 1.0013          & 1.0019            & 1.0034            \\ \cline{2-7} 
\multicolumn{1}{|c|}{}                                                    & Indexing Time (s) & 439.434 (4.04) & 952.773 (8.76)         & 1922.7 (17.67)         & 7519.43 (69.1)          & \textbf{108.782}           \\ \hline
\multicolumn{1}{|c|}{\multirow{4}{*}{\textbf{\makecell[c]{Yandex \\ Deep500M} }}}           & Query Time (ms)   & 28546.6 (1.09) & 61657.9 (2.35)        & 91724.2 (3.50)        & 62411.8 (2.38)          & \textbf{26200.4}           \\ \cline{2-7} 
\multicolumn{1}{|c|}{}                                                    & Recall            & \textbf{0.9852}  & 0.9644          & 0.9298          & 0.9506            & 0.9176            \\ \cline{2-7} 
\multicolumn{1}{|c|}{}                                                    & Overall Ratio     & \textbf{1.0003}  & 1.0009          & 1.0032          & 1.0009            & 1.0058            \\ \cline{2-7} 
\multicolumn{1}{|c|}{}                                                    & Indexing Time (s) & 2263.87 (4.22) & 17182.7 (32.04)        & 13685.2 (25.52)        & 85968.3 (160.3)          & \textbf{536.262}           \\ \hline
\multicolumn{1}{|c|}{\multirow{4}{*}{\textbf{\textbf{\makecell[c]{Microsoft \\ SPACEV500M} } }}}      & Query Time (ms)   & 31404.3 (1.07) & 66632.3 (2.28)        & 94868.3 (3.25)        & 70697.5 (2.42)          & \textbf{29212.6}           \\ \cline{2-7} 
\multicolumn{1}{|c|}{}                                                    & Recall            & \textbf{0.963}   & 0.9492          & 0.9568          & 0.9198            & 0.8978            \\ \cline{2-7} 
\multicolumn{1}{|c|}{}                                                    & Overall Ratio     & \textbf{1.0008}  & 1.0012          & 1.0011          & 1.0026            & 1.00336           \\ \cline{2-7} 
\multicolumn{1}{|c|}{}                                                    & Indexing Time (s) & 2204.94 (4.21) & 16114.7 (30.77)        & 13189.5 (25.19)        & 87591.1 (167.3)          & \textbf{523.662}           \\ \hline
\multicolumn{1}{|c|}{\multirow{4}{*}{\textbf{\makecell[c]{Microsoft \\ Turing-ANNS500M} }}} & Query Time (ms)   & 31280.1 (1.04) & 68636.6 (2.28)        & 106987 (3.55)         & 73618.2 (2.44)          & \textbf{30127.2}           \\ \cline{2-7} 
\multicolumn{1}{|c|}{}                                                    & Recall            & \textbf{0.9806}  & 0.9604          & 0.9636          & 0.9404            & 0.9008            \\ \cline{2-7} 
\multicolumn{1}{|c|}{}                                                    & Overall Ratio     & \textbf{1.0005}  & 1.0012          & 1.0009          & 1.0012            & 1.0043            \\ \cline{2-7} 
\multicolumn{1}{|c|}{}                                                    & Indexing Time (s) & 2301.02 (4.22) & 16408.2 (30.11)        & 12680.2 (23.27)        & 79162.5 (145.3)          & \textbf{545.006}           \\ \hline
	\end{tabular}
 } % font size
%	\begin{tablenotes} 
%		% \item Note that for DET-LSH and DET-ONLY, the time of encoding phase is included in the indexing time.
%        \item \textcolor{red}{The best value in each row is highlighted in bold. The number in parentheses indicates the number of times that this method is slower than the best method.}
%	\end{tablenotes}
\end{threeparttable}
\end{table*}

\subsection{Self-evaluation of DET-LSH} \label{selfevaluation}

% We self-evaluate DET-LSH by counting the specific running time of each algorithm and comparing the performance of algorithms before and after their optimizations.

\subsubsection{Encoding and Indexing Phase} \label{encodeoptimize}

Figure \ref{encodingandindexing} shows the specific running time of each algorithm in the encoding and indexing phases. 
We have the following observations: 
(1) \textit{Dynamic Encoding} (Algorithm \ref{dynamic_encoding}) takes longer time than \textit{Create Index} (Algorithm \ref{create_index}). 
Although we have optimized the process of locating regions when encoding through binary search, 
it still takes much time to locate a specific region from 256 regions for each dimension of each projected point. 
(2) Optimized \textit{Breakpoints Selection} (Algorithm \ref{getbreakpoints}) achieves 3x speedup in running time over the unoptimized algorithm. 
As mentioned in Section \ref{Encoding Phase}, we use \textit{QuickSelect} algorithm with \textit{divide-and-conquer} strategy to avoid complete sorting, 
thus reducing the time complexity from $\mathcal{O}(n\log n)$ to $\mathcal{O}(n\log N_r)$.

\subsubsection{Query Phase} \label{queryoptiize}

In practice, in Algorithm \ref{traversesubtree}, if the upper bound distance between a leaf node and $q^\prime$ is greater than the search radius, it will take much time to calculate the distance between each point in the leaf node and $q^\prime$ (lines 8-13).
After experiments, we found that if the leaf node size $max\_size$ is appropriately set in Algorithm \ref{create_index}, 
most of the points in these \enquote{troublesome} leaf nodes are within the search radius. 
Therefore, we optimized Algorithm \ref{traversesubtree} in two aspects: (1) We relax the requirements for candidate points to improve efficiency. 
In our implementation, as long as the lower bound distance between a leaf node and $q^\prime$ is not greater than $r$, we will add all its points to $S$. 
(2) We maintain a priority queue to hold traversed leaf nodes based on their lower bound distances to $q^\prime$. 
A leaf node with a smaller lower bound distance to $q^\prime$ can add all its points to $S$ earlier, guaranteeing the quality of candidate points. 
As shown in Figure \ref{querytime}, with an acceptable sacrifice of query accuracy, optimized Algorithm \ref{range_query} and Algorithm \ref{ckann} improve query efficiency by up to 50\% and 30\%.

\begin{figure*}[tb] 
	\centering
	\includegraphics[width=0.92\linewidth]{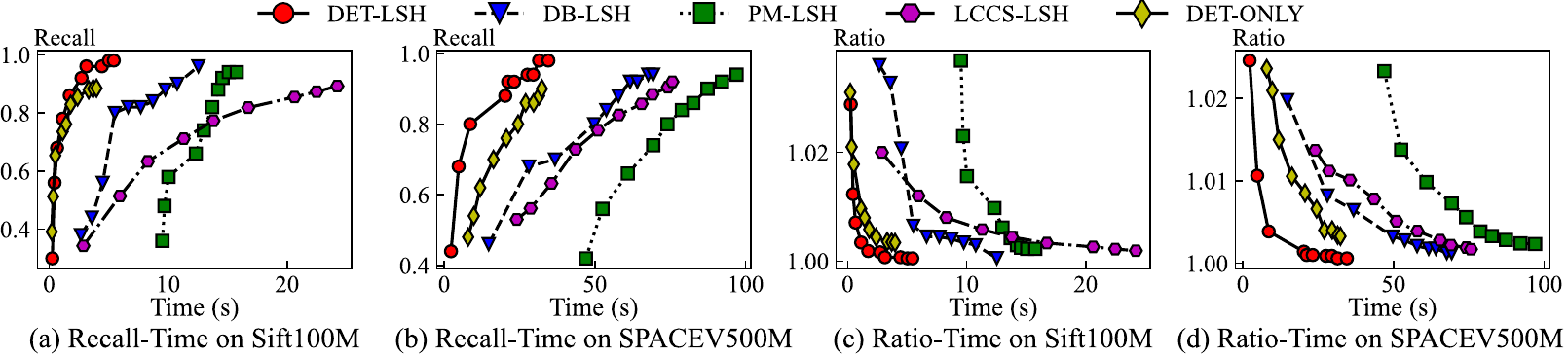}
	\caption{Recall-time and overall ratio-time curves.}
	\label{recalltimeandratiotime}
\end{figure*}

\begin{figure*}[tb] 
	\centering
	\includegraphics[width=0.92\linewidth]{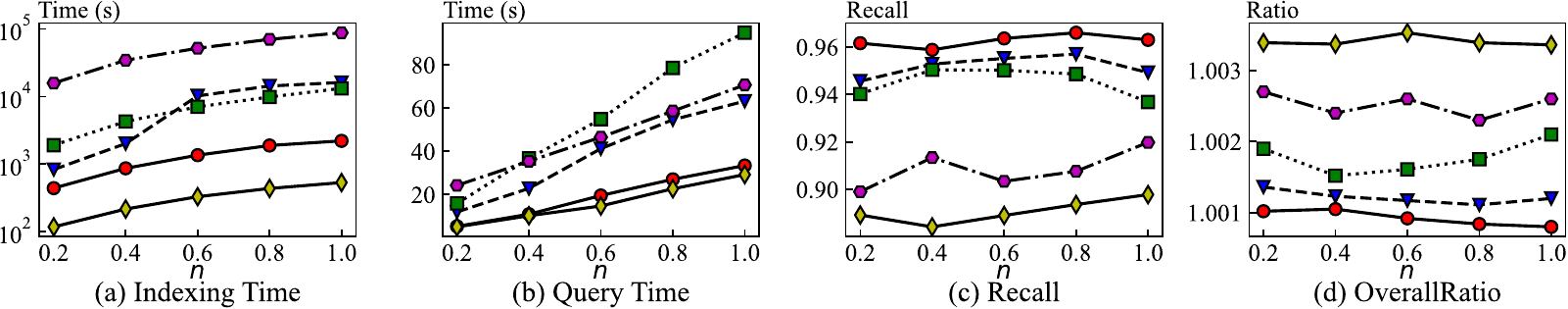}
	\caption{Scalability: performance under different $n$ on Microsoft SPACEV500M.}
	\label{scalability}
\end{figure*}

\begin{figure*}[tb] 
	\centering
	\includegraphics[width=0.92\linewidth]{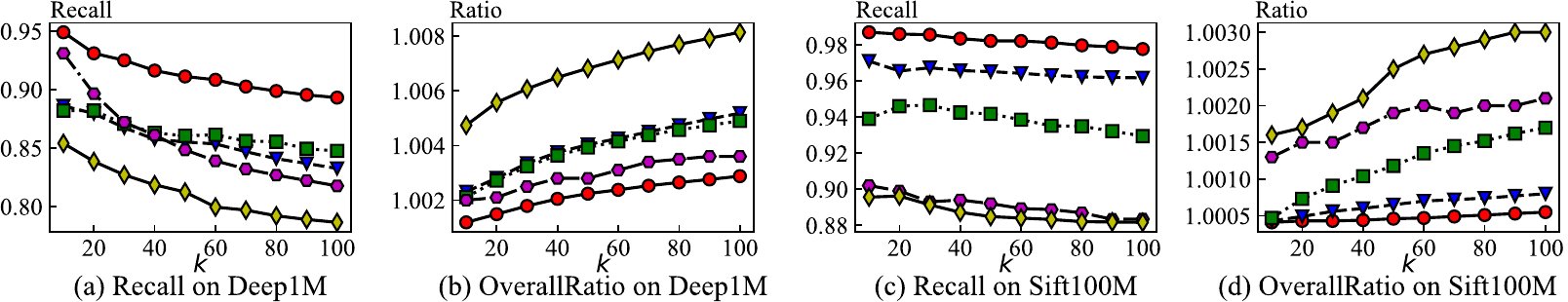}
	\caption{Performance under different $k$.}
	\label{diffk}
\end{figure*}

\subsection{Comparison with Competitors}

\subsubsection{Indexing Performance} 
Figure \ref{indexsize} and Table \ref{table3} show the comparison between all methods with default parameter settings on all datasets.
To ensure fairness, for DET-LSH and DET-ONLY, the time of the encoding phase is included in the indexing time. 
We make the following observations: 
(1) DET-LSH has the best indexing efficiency compared to all LSH-based methods. 
The reason is that DB-LSH and PM-LSH use data-oriented partitioning trees to construct indexes. 
It is time-consuming to partition a multi-dimensional projected space. 
DET-LSH adopts DE-Trees to construct indexes, which divide and encode each dimension of the projected space independently, thereby improving indexing accuracy. 
LCCS-LSH has a significantly longer indexing time compared to other methods because of building its proposed data structure Circular Shift Array (CSA). 
(2) The advantage of DET-LSH's indexing efficiency increases with the dataset cardinality.
When $n$ is not greater than 1M, the indexing time of DET-LSH is longer than that of PM-LSH, because DET-LSH constructs 4 DE-Trees, 
while PM-LSH only constructs one PM-Tree. 
As $n$ increases from 10M to 500M, DET-LSH achieves from 2x speedup to 6x speedup in indexing time over other methods. 
The reason is that the construction time of a DE-Tree increases linearly with $n$.
(3) With respect to index size, DET-LSH is not very competitive on small-scale datasets, but its design proves advantageous for large-scale datasets. 
The reason is that DET-LSH only saves the iSAX representation of each data point in the DE-Tree (not the original data point or the LSH-projected data point). 
Each iSAX representation is stored as an \enquote{unsigned char}, which takes up only one byte. 
Yet, DET-LSH builds 4 DE-Trees, which restricts its advantage on small-scale datasets.
As the dataset size increases, the advantage of DET-LSH becomes more pronounced.
(4) The index size and indexing time of DET-ONLY are always about one-quarter of that of DET-LSH. 
The reason is that DET-ONLY only constructs one DE-Tree, while DET-LSH constructs 4 DE-Trees. 
% To sum up, DET-LSH has great advantages in indexing efficiency, especially on very large-scale datasets.

\subsubsection{Query Performance}

We study query performance based on the query time, recall, and overall ratio shown in Table~\ref{table3}, and the Recall-Time and OverallRatio-Time curves shown in Figure~\ref{recalltimeandratiotime}. 
We have the following observations: (1) DET-LSH outperforms all LSH-based methods on both efficiency and accuracy 
(DET-ONLY is not an LSH-based method). 
As shown in Table~\ref{table3}, DET-LSH has a shorter query time, higher recall, and smaller ratio on all datasets. 
As $n$ increases, DET-LSH achieves up to 2x speedup in query time over other LSH-based methods. 
The reason is that closer points have similar encoding representations in DE-Tree so that range queries can obtain higher-quality candidate points in a shorter time. 
(2) The query efficiency of DET-ONLY is slightly better than DET-LSH, but the query accuracy is significantly lower than DET-LSH. 
Since DET-LSH performs queries on 4 DE-Trees, querying is more expensive in terms of time cost, but more accurate than DET-ONLY, which uses a single DE-Tree to answer queries. 
DET-LSH can control the trade-off between accuracy and efficiency by adjusting the number of DE-Trees, while DET-ONLY cannot do this.
The performance of DET-ONLY shows that it is not suitable to support accurate ANN queries, demonstrating the importance of the of the LSH component to guarantee query accuracy.
Overall, DET-LSH is more advantageous than DET-ONLY.
(3) DET-LSH achieves the best trade-off between efficiency and accuracy. As shown in Figure~\ref{recalltimeandratiotime}, 
compared with other LSH-based methods, DET-LSH consumes the least time to achieve the same recall or overall ratio.

\begin{figure}[tb] 
	\centering
	\includegraphics[width=0.94\linewidth]{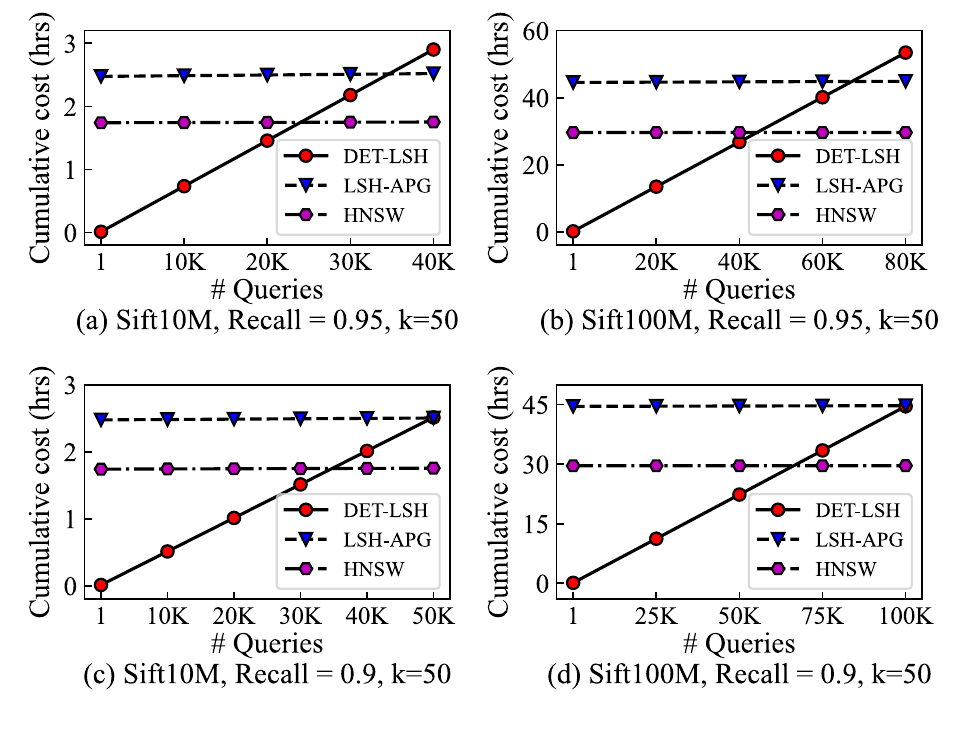}
	\caption{Cumulative query cost (first query includes indexing time).}
	\label{lshvsgraph}
\end{figure}

\subsubsection{Scalability} 

A method has good scalability if it performs well on datasets of different cardinalities. 
To investigate the scalability of all methods, we randomly select different number of points from the \textit{Microsoft SPACEV500M} dataset and compare the indexing and query performance of all methods under default parameter settings. 
Figure~\ref{scalability} shows the results. 
We have the following observations: 
(1) Although the indexing and query times increase with the cardinality for all methods, DET-LSH grows much slower than other LSH-based methods due to the efficiency of DE-Tree (Figure~\ref{scalability}(a) and Figure~\ref{scalability}(b)). 
DET-ONLY constructs indexes and answers queries faster, but the accuracy is much less than other methods.
(2) The recall and overall ratio are relatively stable for all methods. 
The reason is that the data distribution does not change significantly with the cardinality because we select points randomly. 
To sum up, DET-LSH has better scalability than other LSH-based methods.

\subsubsection{Effect of $k$} 
To investigate the effect of $k$, we evaluate the performance of all methods under different $k$. 
Since changing $k$ has little impact on query time, and has no impact on indexing time, we only report the results on recall and overall ratio, shown in Figure~\ref{diffk}. 
We make the following observations: (1) As $k$ increases, the query accuracy of all methods decreases slightly. 
The reason is that the number of candidate points does not change with $k$. A larger $k$ means it is more likely to miss the exact NN points. 
(2) DET-LSH consistently exhibits the best performance among all competitors.

\subsection{Comparison with Graph-based Methods}

% \begin{figure}[tb] 
% 	\centering
    % \includegraphics[width=0.4\linewidth]{./figures/lshvsgraphindexing.pdf}
    % \caption{\textcolor{red}{Index size.}}
    % \label{lshvsgraphindex}
% \end{figure}

% \begin{figure}[tb] 
% 	\centering
% 	\includegraphics[width=0.4\linewidth]{./figures/updateefficiency.pdf}
% 	\caption{\textcolor{red}{Update efficiency.}}
% 	\label{updateefficiency}
% \end{figure}

In this section, we compare DET-LSH to graph-based methods~\cite{malkov2018efficient,fu2019fast,peng2023efficient,azizi2023elpis,wang2021comprehensive}. 
Nevertheless, LSH-based and graph-based methods have different design principles and characteristics~\cite{hydra2,li2019approximate,zeyubulletin-sep23}, making them suitable to different application scenarios. 
In particular, graph-based methods only support ng-approximate answers~\cite{hydra2}, that is, they do not provide any quality guarantees on their results.
It is important to emphasize that DET-LSH has to pay the cost of providing guarantees for its answers; graph-based methods, that do not provide any guarantees, do not pay this cost.

In our previous experiments, we demonstrated that DET-LSH outperforms other LSH-based methods. 
In this section, we compare DET-LSH to HNSW~\cite{malkov2018efficient}, the state-of-the-art graph-based method~\cite{hydra2}. 
In addition, we compare to a hybrid method, LSH-APG~\cite{lshapg}, which uses LSH to retrieve a high-quality entry point for the subsequent search in an Approximate Proximity Graph (APG). 
%HNSW is the most widely used graph-based method, which constructs navigable small world graphs with controllable hierarchy to support efficient queries at high recall.
%HNSW does not provide any quality guarantees for the answers, while our method DET-LSH has to pay the cost of providing guarantees for its answers, which gives an unfair advantage to HNSW. 

\begin{figure}[tb]
    \centering
    \begin{minipage}[t]{0.49\linewidth}
        \includegraphics[width=0.85\linewidth]{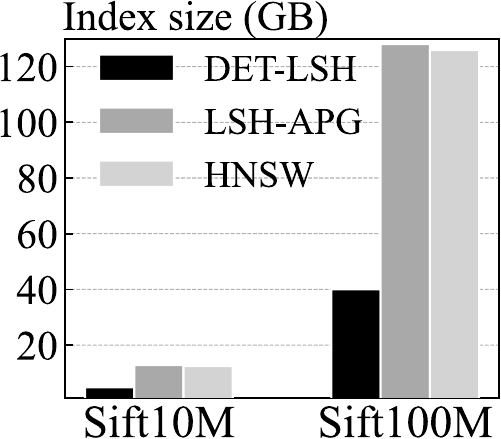}
	\caption{Index size.}
	\label{lshvsgraphindex}
    \end{minipage}
    \begin{minipage}[t]{0.49\linewidth}
        \includegraphics[width=0.95\linewidth]{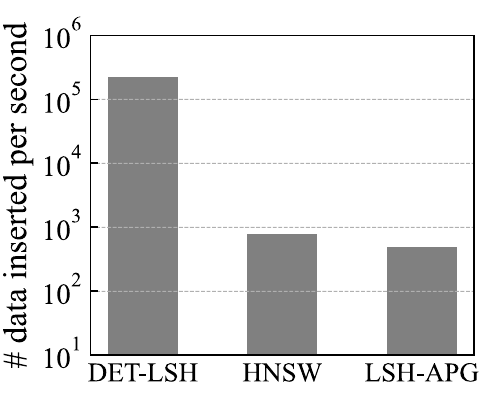}
	\caption{Update efficiency.}
	\label{updateefficiency}
    \end{minipage}
\end{figure}

In terms of indexing and query efficiency, Figure~\ref{lshvsgraph} shows the cumulative query costs of DET-LSH, HNSW, and LSH-APG, where the cost of the first query also includes the indexing time. 
We observe that, as expected, DET-LSH has an advantage in indexing efficiency: it creates the index and answers 30K-70K queries before the best competitor (i.e., HNSW) answers its first query. 
This behavior is partly explained by the more succinct index structure of DET-LSH. 
Figure~\ref{lshvsgraphindex} shows that the DET-LSH index is almost 3x smaller in size than the index constructed by the competitors.
Finally, Figure~\ref{updateefficiency} shows the update efficiency of these methods, by measuring the number of data points per second when inserting the last 10M points of the \textit{Sift100M} dataset into the existing indexes.
In this scenario that involves updates, DET-LSH is 2-3 orders of magnitude faster than HNSW and LSH-APG. % has great advantages in update efficiency. DET-LSH can insert about 223K points per second, while HNSW and LSH-APG can only insert about 0.8K and 0.5K points per second. 

% Graph-based methods have become another efficient solution for ANN search. 
% Compared with LSH-based methods, graph-based methods typically have better query efficiency, but require considerably more time to construct indexes. 
% The reason is that graph-based methods need to identify for each data point in the dataset its near neighbors (and connect to them) during the indexing phase, while in the query phase, they only need to search on a gradually converging path. 

In summary, LSH-based methods (DET-LSH) have distinct characteristics and different advantages when compared to graph-based methods, pure (such as HNSW) or hybrid (such as LSH-APG), making each method better suited for different scenarios. 

\section{Conclusions} \label{chapter7}

In this paper, we have proposed a novel LSH scheme, called DET-LSH, to efficiently and accurately answer $c$-ANN queries in high-dimensional spaces with strong theoretical guarantees. 
DET-LSH combines the ideas of BC and DM methods, constructing multiple index trees to support range queries based on the Euclidean distance metric, which reduces the probability of missing exact NN points and improves query accuracy. 
To efficiently support range queries in DET-LSH, we designed a dynamic encoding-based tree called DE-Tree, which outperforms data-oriented partitioning trees used in existing LSH-based methods, especially in very large-scale datasets. 
Extensive experiments demonstrate that DET-LSH outperforms the state-of-the-art LSH-based methods in both efficiency and accuracy.

\begin{acks}
This work is supported by the National Natural Science Foundation of China (NSFC) under the grant number 62202450, and supported by the Hellenic Foundation for Research and Innovation (HFRI)
under the “Second Call for HFRI Research Projects to support Faculty
Members and Researchers” (project number: 3684).
\end{acks}

%\clearpage

\bibliographystyle{ACM-Reference-Format}
\balance
\bibliography{ref}

\end{document}